\documentclass[
    11pt,
    aps,
    pra,
    reprint,
    twocolumn,
    notitlepage,
    nofootinbib,
    longbibliography
]{revtex4-2}
\usepackage[T1]{fontenc}
\usepackage{amsmath,amssymb,amsthm,mathtools,bm}
\usepackage{graphicx,booktabs,array,xcolor,tikz}
\usepackage[normalem]{ulem}
\usepackage{pifont}

\definecolor{niceblue}{HTML}{2B5A8A}

\newenvironment{behaviorarray}[1]
{\setlength{\arraycolsep}{3pt}\begin{array}{#1}}
{\end{array}}
\usetikzlibrary{arrows.meta, positioning, calc, hobby, shapes, decorations.pathmorphing}
\usepackage[
    colorlinks=true,
    linkcolor=blue!45!black,
    citecolor=blue!45!black,
    urlcolor=blue!55!black
]{hyperref}
\newcommand{\GW}{\mathcal{GW}}
\newcommand{\IC}{\mathcal{IC}}
\newcommand{\Q}{\mathcal Q}
\newcommand{\Qc}{\mathcal Q_{\mathrm c}}
\newcommand{\LHV}{\mathcal L}
\newcommand{\NS}{\mathcal{NS}}
\newcommand{\Rvec}{\mathbf R}
\newcommand{\rvec}{\mathbf r}
\newcommand{\Bvec}{\mathbf B}
\newcommand{\bvec}{\mathbf b}
\newcommand{\Bghost}{\tilde{\mathbf B}}
\newcommand{\DKL}{D_{\mathrm{KL}}}

\newcommand{\conv}{\operatorname{conv}}
\newcommand{\tr}{\operatorname{tr}}

\newcommand{\beq}{\begin{equation}}
\newcommand{\eeq}{\end{equation}}
\newtheorem{theorem}{Theorem}
\newtheorem{result}{Result}
\newtheorem{proposition}{Proposition}
\newtheorem{lemma}{Lemma}
\newtheorem{corollary}{Corollary}
\theoremstyle{definition}
\newtheorem{definition}{Definition}

\allowdisplaybreaks[2]
\usepackage{lipsum}

\begin{document}
\title{Information causality does not single out the set of quantum correlations}
\author{Marco Pompili}
\affiliation{Inria, Télécom Paris -- LTCI, Institut Polytechnique de Paris, 91120
    Palaiseau, France}
\author{Mirjam Weilenmann}
\affiliation{Inria, Télécom Paris -- LTCI, Institut Polytechnique de Paris, 91120
    Palaiseau, France}
\begin{abstract}
Information causality (IC) was introduced as a physical principle that imposes constraints on the non-local correlations between two parties (Alice and Bob) in a Bell scenario. Specifically, it states that the information that Bob can gain about an $n$-bit database held by Alice is bounded by the information she communicates to him. To date, the question whether IC for any $n \in \mathbb{N}$ singles out the set of quantum correlations remains open, meaning that it is unknown whether for any number of inputs $m$ and outputs $d$ of the two parties the set of correlations that satisfy IC exactly coincides with the quantum correlations. In this work, we answer this question in the negative. Firstly, we show that for any finite database length $n \geq 2$, there is a bipartite set of correlations that satisfies IC for all compatible protocols and that exhibits post-quantum correlations in the Bell scenario with $n+1$ inputs and $2$ outcomes per party. This result relies on constructing an entropic proof of IC for GPTs that may be of independent interest. We further show that in the Bell scenario with $n+1$ inputs and $2$ outcomes per party, our constructed sets of correlations lead to neither a superset nor a subset of the quantum correlations. Secondly, we construct a set of correlations in the Bell scenario with $3$ inputs and $2$ outputs per party that satisfies IC (for all database sizes $n \in \mathbb{N}$) and that includes post-quantum points. Consequently, all known formulations of information causality fail to fully characterize the bipartite set of quantum correlations, already in the $3$-input, $2$-outcome Bell scenario. 
\end{abstract}
\maketitle

\section{Introduction}

Initiated by Bell's theorem, which first showed that quantum theory features
non-local correlations~\cite{bell:1964}, understanding these correlations has
become a research field in its own right. In their seminal
paper~\cite{popescu:1994}, Popescu and Rohrlich showed that the impossibility to
send superluminal signals alone does not serve as an underlying principle to
explain quantum correlations: indeed, there are non-quantum non-local
correlations that would also be compatible with this
principle~\cite{popescu:1994, cirelson:1993}. Following their work, much effort
has been devoted to the search for information-theoretic physical principles
that imply that non-local correlations in Bell scenarios are restricted to the
quantum set~\cite{brassard:2006, linden:2007, navascues:2010, fritz:2013}. Among
these, information causality stands out as the only principle for which it has
so far not been shown that it is also satisfied by some non-quantum non-local
correlations~\cite{navascues:2015}.

Information causality (IC)~\cite{pawlowski:2009} is an information-theoretic
principle proposed in a scenario where Alice holds a database of $n$ independent
random dits, and Bob is asked to guess one chosen entry. IC requires that the
sum of the information Bob can obtain about the individual entries does not
exceed the information communicated by Alice, even when they share nonlocal
resources. Both classical and quantum theory satisfy this requirement, while
other generalized probabilistic theories (GPTs) such as
\emph{boxworld}~\cite{barrett:2007} do not.

Since its original introduction, various versions of the IC principle have been
formulated~\cite{chaves:2015c, miklin:2021, jain:2024, yu:2026, pollyceno:2026, miklin:2026}.
Recently, it was shown that one of these formulations~\cite{jain:2024},
already with $n = 2$ bits and a binary channel leads, in addition to Tsirelson's
bound~\cite{pawlowski:2009}, also to the Tsirelson--Landau--Masanes criterion
that characterizes the quantum correlator region for two binary measurements per
party~\cite{gachechiladze:2026}. Nevertheless, the general question whether (any
of the formulations of) IC recovers the set of quantum correlations remains open.

In this work, we show that information causality, in all formulations known to
us~\cite{pawlowski:2009,chaves:2015c,miklin:2021,jain:2024,miklin:2026,pollyceno:2026,yu:2026}, does not single out the bipartite quantum set of correlations. We
show this by constructing sets of correlations that feature non-quantum nonlocal
correlations and prove that these sets satisfy IC for any $n \in \mathbb{N}$.

Our proofs rely on sets of correlations that allow for symmetric non-signalling
extensions. Specifically, we view the \emph{ghostworld} sets of correlations,
introduced in  \cite{lang:2014}, as bipartite state spaces of GPTs and define corresponding effect spaces.
 Using a notion of conditional
measurement entropy for generalised probabilistic theories~\cite{barnum:2010,
short:2010a, weilenmann:2020}, we prove that for any finite database size $n$
there is a bipartite set of correlations that satisfies IC and that features non-quantum non-local
correlations (Result~\ref{res:gw-ic}). 
We obtain this by providing a proof of IC for GPT correlations, alternative to the well-known one from Refs.~\cite{barnum:2010, al-safi:2011}, which may be of independent
interest.

To generalize this to the arbitrary $n \in \mathbb{N}$ case in a fixed Bell scenario, we construct
a smaller set of correlations $\mathcal{S}$ (unrelated to GPTs) in the setting with three binary measurements per party, for which we show that it
satisfies IC for all $n \in \mathbb{N}$. We first
show this in the restricted case where Alice and Bob only have access to
one copy of their shared non-local resource, which we show solves the problem
for the formulation of IC in  \cite{miklin:2021}
(Result~\ref{res:contraction}). We then extend this result to an unbounded
number of copies (Result~\ref{res:contraction-extension}), thus establishing the
general result. The proof techniques are built on previous results on
(hyper)contractivity from information theory~\cite{beigi:2015, anantharam:2013b, anantharam:2013, raginsky:2016},
applied to the study of non-local correlations.

\section{Problem statement and summary of main results}
\label{sec:IC_and_quantum_sets}

Let us consider two parties, Alice and Bob, who share a bipartite 
system, on which they can perform measurements labelled by $x, y \in [m]:=\{0, \ldots, m-1 \}$ and obtain outcomes $a, b \in [d]:=\{0, \ldots, d-1 \} $ respectively. Any choice of resource and measurements leads to a correlation $P=\{ p(a,b | x,y)\}_{a,b,x,y}$ in the so-called $(m,m,d,d)$ Bell scenario, where we will suppress the subscripts ${a,b,x,y}$ throughout this article. We can generally think of such a correlation, also called \emph{behavior}, as a \emph{nonlocal box}, where the parties make inputs $x,y$ and obtain outcomes $a,b$ respectively.

Depending on the resources the parties share, different sets of correlations are possible. 
The quantum set $\Q^{(m,d)}$ is defined as the set of correlations for
which there exist separable Hilbert spaces $\mathcal H_A$ and $\mathcal H_B$,
POVMs $\{A_{a|x}\}_a$, $\{B_{b|y}\}_b$ indexed by $x, y$ on them, respectively, and a state
$\rho$ on $\mathcal H_A \otimes \mathcal H_B$ such that
$p(a,b|x,y) = \tr(\rho A_{a|x} \otimes B_{b|y})$. We denote its topological closure as $\overline{\Q}^{(m,d)}$. The quantum commuting set of
correlations $\Qc^{(m,d)}$ requires a single Hilbert space and
$p(a,b|x,y) = \tr(\rho A_{a|x} B_{b|y})$, with $[A_{a|x}, B_{b|y}] = 0$ for all
$a,b,x,y$. It is well-known that for all $m,d$ 
\beq
    \label{eq:icq-hierarchy}
    \Q^{(m,d)} \subseteq \overline{\Q}^{(m,d)}
       \subseteq \Qc^{(m,d)},
\eeq
and that each inclusion is strict for at least one choice of $(m,d)$; 
the strictness of the first inclusion was proved in~\cite{slofstra:2019} and of the second in~\cite{ji:2021}. 

Information-theoretic principles, similarly, define sets of correlations satisfying them. Here we consider the specific case of information causality~\cite{pawlowski:2009}, see Figure~\ref{fig:ic_notation}, where Alice holds a database
$\Rvec = (R_0,\dots,R_{n-1})$ of independent, uniformly distributed random
variables taking values in $[k]$.\footnote{Independence is assumed here for simplicity of presentation. The technical parts of the manuscript in fact apply also to the formulation of IC with correlated databases~\cite{chaves:2015c, jain:2024} (see Sections~\ref{sec:insufficiency-finite-length}--\ref{sec:insufficiency-general}).} Bob receives an
independent, uniformly distributed index $S \in [n]$, with the goal of guessing
$R_S$. To help them perform this task, the parties share nonlocal resources from some theory, e.g. quantum systems or systems from a GPT. The parties can perform local measurements on their respective parts of these systems; we can thus think of them as generating nonlocal boxes in the process. We will in particular be interested in the case where the parties share an arbitrary number of such boxes from the $(m,m,d,d)$ Bell scenario. In the quantum case this means any number of (possibly different) $P \in \Qc^{(m,d)}$, for a GPT any number of boxes $P$ that can be generated according to that theory in the $(m,m,d,d)$ Bell scenario. In this case, Alice can choose her measurement based on the database $\bf R$ and on the choices and outcomes of any previous measurements she has performed. In addition, she can send a message $M'$ to Bob. This message can in particular depend on the database $\bf R$ and on Alice's inputs $\mathbf X$ and outcomes $\mathbf A$ of each box.
While the original protocol considers a perfect channel to transmit the message, i.e. $M=M'$~\cite{pawlowski:2009}, more recent formulations~\cite{miklin:2021} consider a noisy channel whose output $M' \in [d']$ is received by Bob. 
Bob's measurement choices may depend on $S$, $M'$, and possibly on the inputs and outputs of the measurements he already performed.
He then produces a guess $G \in [k]$ using this information and his measurement outcomes. 

\begin{figure}
    \centering
    \includegraphics{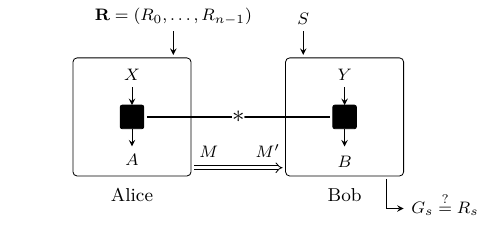}
    \caption{\textbf{Notation.} A $n$-dit random-access code assisted by non-signaling resources, with a noisy classical communication $M \to M'$.} %
    \label{fig:ic_notation}
\end{figure}

We say that a set $\mathcal S$ of nonlocal boxes (containing the local set) in the $(m,m,d,d)$ Bell scenario satisfies information causality with database size $n$, if any number of independent boxes from this set satisfy
\beq
    \label{eq:ic_inequality_noisy}
    \sum_{s=0}^{n-1}I(R_s:G_s) \leq I(\Rvec,\mathbf X, \mathbf A : M'),
\eeq
for all $[d']$ and all protocols, where
$I(R_s:G_s) \coloneqq I(R_s:G \mid S = s)$ and $I$ denotes the Shannon mutual
information. Notice that the potential use of shared randomness by the parties in a protocol is incorporated here in the set of resources (i.e.\ accounted for within $\mathbf A$).
Similarly, for a behavior $P$ in the $(m,m,d,d)$ Bell scenario, we write $P \in \IC_{n}^{(m,d)}$ if any number $c \in \mathbb N$ of independent copies of the behavior $P$ satisfy Eq.~\eqref{eq:ic_inequality_noisy}
for a database of size $n$, while $\IC_{\infty}^{(m,d)}$ refers to the case of arbitrary $n \in \mathbb{N}$. The inequality in Eq.~\eqref{eq:ic_inequality_noisy} will be particularly convenient for our proofs. Since any message $M$ that Alice may send is a function of $\bf R, \bf X, \bf A$, we have the Markov chain $\mathbf R, \mathbf X, \mathbf A \rightarrow M \rightarrow M'$, which implies that
\beq
I(\Rvec,\mathbf X, \mathbf A :M') \leq I(M:M') \leq H(M), 
\eeq
see also~\cite{marco:costantino}. Proving that~\eqref{eq:ic_inequality_noisy} is satisfied thus immediately implies that also the usual IC notions from~\cite{pawlowski:2009, miklin:2021} are satisfied. 
The requirement that the IC inequality holds for all sets of behaviors from a given theory,
under any protocol, together with the fact that IC is formulated in terms of an entropic constraint, explains the difficulties in analysing this principle even for moderate $(m,d)$ and $n$, and even more so for arbitrary $n \in \mathbb{N}$.

\begin{figure*}
    \centering
    \includegraphics{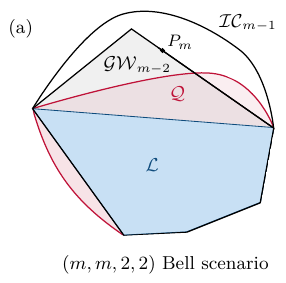} \qquad \qquad
    \includegraphics{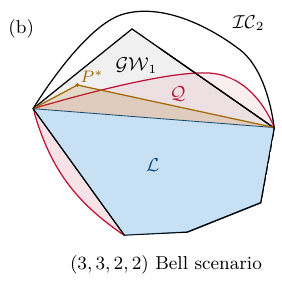}
    \caption{
        \textbf{Sets of correlations considered in our results.} (a) In the Bell
        scenarios with $m$ inputs and $2$ outputs per party,
        $\GW_{m-2}^{(m,2)} \subseteq \IC_{m-1}^{(m,2)}$, and $\GW_{m-2}^{(m,2)}$ 
        contains post-quantum behaviors, like $P_m$ achieving $I_{mm22} = 1/m$, for all $m \ge 3$
        (Result~\ref{res:gw-ic}). (b) In the Bell scenario with $3$ inputs and
        $2$ outputs per party, there exists a behavior $P^* \in \IC_\infty^{(3,2)}$ but outside
        $\Q^{(3,2)}$, such that any protocol involving correlations in
        $\mathcal S^{(3,2)} = \conv(\LHV^{(3,2)} \cup \{P^*\})$ cannot violate IC for any number
        of systems or database length (Results~\ref{res:contraction}
        and~\ref{res:contraction-extension}).
    }
    \label{fig:results}
\end{figure*}

The question whether IC with database size $n$ recovers the set of 
quantum correlations in Bell experiments amounts to asking whether
$\IC_n^{(m,d)}=\Qc^{(m,d)}$ for all $m,d$, or, alternatively, whether
$\IC_n^{(m,d)}=\overline{\Q}^{(m,d)}$ for all $m,d$. In principle, 
this must hold for any $(m_A,m_B,d_A,d_B)$ Bell scenario; however
for the purpose of this article the $(m,m,d,d)$ case suffices, to which we thus restrict for ease of presentation. Note further that,
since $\IC_n^{(m,d)}$ is closed, it cannot coincide
with $\Q^{(m,d)}$ for $m,d$ where the latter is not closed~\cite{slofstra:2019}.

In our work, we identify sets of correlations $\mathcal{S}^{(m,d)}$ that satisfy $\IC_n^{(m,d)}$, i.e., $\mathcal{S}^{(m,d)} \subseteq \IC_n^{(m,d)}$ and that feature behaviours $P \in \mathcal{S}^{(m,d)}$ that lie outside $\Qc^{(m,d)}$, thus solving the question for both types of quantum sets.\footnote{One could also ask, more generally, whether a physical principle recovers \emph{all} multipartite quantum correlations. However, bipartite principles are
known to be generally insufficient for recovering all  multipartite correlations~\cite{gallego:2011}.} Whenever inclusion is achieved for all $m,d$ we will denote this $\mathcal{S} \subseteq \IC_n$.

To solve this question, we study the relation between IC and
the \emph{ghostworld} sets of correlations, introduced in~\cite{lang:2014} as sets of symmetric non-signaling extendible behaviors, a
concept analogous to the well-known Doherty--Parrillo--Spedalieri hierarchy for
entanglement detection~\cite{doherty:2002, doherty:2004}.

\begin{definition}[Ghostworld correlations]
    \label{def:gw}
    A bipartite behavior $P = \{p(a,b|x,y)\}$ belongs to $\GW_n$ if there exist
    two $(n + 2)$-partite non-signaling behaviors
    $Q_A = \{q_A(\mathbf a,b|\mathbf x,y)\}$ and
    $Q_B = \{q_B(a,\bvec|x,\mathbf y)\}$ satisfying
    \beq
        \label{eq:gw-explicit-sums-alternative}
        \begin{aligned}
            \sum_{a_1,\ldots,a_n} q_A(\mathbf a,b|\mathbf x,y) & = p(a_0,b|x_0,y), \\
            \sum_{b_1,\ldots,b_n} q_B(a,\bvec|x,\mathbf y)     & = p(a,b_0|x,y_0).
        \end{aligned}
    \eeq
    Here, $\mathbf a = (a_0,\ldots,a_n)$, and analogously for
    $\bvec,\mathbf x,\mathbf y$, and $Q_A$ and $Q_B$ are invariant under
    permutations of the pairs $(a_i,x_i)$ and $(b_i,y_i)$, respectively, with
    $i = 0,\ldots,n$.
\end{definition}

The sets of behaviours $\GW_n^{(m,d)}$ are convex polytopes for each $m,d$. While it was known that $\GW_1^{(3,2)}$ features non-quantum non-local behaviours~\cite{lang:2014}, we show by analyzing for each $n$ the optimal $\GW_n$ violation of the Collins--Gisin $I_{(n+2)(n+2)22}$
Bell inequality~\cite{collins:2004a}, that non-quantum points exist for each $n$. Furthermore, it is known that for any $n \geq 1 $, $\GW_n$ cannot violate any two-input Bell inequalities, such as e.g.\ CHSH~\cite{lang:2014} and its liftings to higher Bell scenarios, which means that there are also $P\in \Q_c \setminus \GW_n$.

Since the set of $\GW_n^{(m,d)}$ correlations is closed under wirings for each $n,m,d$~\cite{lang:2014}, we were furthermore able to understand the sets $\GW_n$ as bipartite state spaces of a GPT. Using conditional measurement entropies for GPTs~\cite{barnum:2010, short:2010a, weilenmann:2020}, we prove the following first main result.

\begin{result}[Insufficiency of finite-length IC]
    \label{res:gw-ic}
    For all finite $n \ge 2$ and any Bell scenario, $\GW_{n-1} \subseteq \IC_n$. Furthermore, for all finite $n \ge 2$ there
    exists a behavior $P_{n+1} \in \GW_{n-1}^{(n+1,2)}$ such that
    \beq
        P_{n+1} \in \IC_n^{(n+1,2)} \setminus \Qc^{(n+1,2)}.
    \eeq
    In other words, information causality for any finite number  $n$ of dits does not
    single out the quantum set.
\end{result}

The details of our construction and proofs are provided in
Section~\ref{sec:insufficiency-finite-length} below.

To consider the case of arbitrary $n \in \mathbb{N}$, we take a subset of the ghostworld correlations: $\mathcal S^{(3,2)} \coloneqq \conv(\LHV^{(3,2)} \cup \{P^*\}) \subsetneq \GW_1^{(3,2)}$, where $\LHV^{(3,2)}$ is the set of local correlations and $P^*$ is a specific behaviour. 
Using tools from~\cite{anantharam:2013, anantharam:2013b, polyanskiy:2025}, we show that some explicit behaviours including $P^*$ satisfy a
fixed-input strong data-processing inequality~\cite{raginsky:2016, polyanskiy:2025} that gives a sufficient
condition for IC to hold for a single copy of any behaviour in $\mathcal S^{(3,2)}$. These
arguments are presented in Sec.~\ref{sec:insufficiency-single-copy} and they
yield the following. %

\begin{result}[Insufficiency of single-copy IC]
    \label{res:contraction}
    There exists a set of correlations $\mathcal S^{(3,2)}$ such that when restricted to using a single copy of the behaviour, $\mathcal S^{(3,2)} \subseteq \IC_\infty^{(3,2)}$. Furthermore there is a $P^* \in \mathcal S^{(3,2)}$ such that
    \beq
        P^* \notin \Qc^{(3,2)}. %
    \eeq
    In other words, single-copy information causality in the noisy channel
    formulation~\cite{miklin:2021} does not single out the quantum set.
\end{result}

Finally, we extend the argument to arbitrary multi-copy protocols
using any finite number of independent copies of the same behavior,
as we discuss in Section~\ref{sec:insufficiency-general}. 

\begin{result}[Insufficiency of general IC]
    \label{res:contraction-extension}
    There exists a set of correlations $\mathcal S^{(3,2)} \subseteq \IC_\infty^{(3,2)}$ and a behavior $P^* \in \mathcal S^{(3,2)}$ such that
    \beq
        P^* \notin \Qc. %
    \eeq
    In other words, the standard definition of information causality does not single out the quantum set~\cite{pawlowski:2009}.
\end{result}

 Notice that all of our results also apply to notions of IC, where the database of Alice is correlated, as used in~\cite{jain:2024, jain:2026,
miklin:2026, gachechiladze:2026} (see technical Sections~\ref{sec:insufficiency-finite-length}--\ref{sec:insufficiency-general}).
In Appendix~\ref{app:other-IC-formulations}, we further show
that our results also hold for all other notions of IC known to us, namely those from~\cite{chaves:2015c, pollyceno:2026, yu:2026}. Our three main results are summarized in Fig.~\ref{fig:results}.

\section{Preliminaries}
\label{sec:preliminary}

In this section we provide the necessary background for the proofs of our three main results.

\subsection{Generalised probabilistic theories and measurement entropy}
\label{sec:prliminary_gw}
In order to understand the relation between IC and the ghostworld sets of correlations from
Definition~\ref{def:gw}, we regard them as state spaces of a family of bipartite
GPTs indexed by $n$. The advantage of this treatment is that we can prove properties of these theories (such as they satisfy IC) independently of the system size and thus for all $m,d$ simultaneously.

A GPT is a convex-geometric framework encompassing
classical probability theory and quantum theory as special cases~\cite{hardy:2001, barrett:2007}, which describes a system by a convex set $\Omega$
of normalized states and a set $\mathcal E$ of effects, i.e.~affine maps
$e\colon \Omega \to [0,1]$. A measurement is a collection of effects $\{e_k\}_k$
summing to the unit effect $u$ that satisfies $u(\omega) = 1$ for all $\omega$. The
probability of outcome $k$ on state $\omega$ is $p(k) = e_k(\omega)$. A
measurement is fine-grained if its nonzero effects are indecomposable: each can
only be written as a sum of positive effects proportional to itself, analogously
to rank-1 POVMs in quantum theory. We write $\mathcal M$ for the allowed
measurements and $\mathcal M^*$ for the fine-grained ones. 

A common example of GPT, besides classical probability theory and quantum theory, is boxworld~\cite{barrett:2007}, where local states are described by the behaviour $\omega_A = \{p(a|x)\}_{a,x}$, where the only restriction is that $\{p(a|x)\}_{a}$ are probability distributions for all $x$. In this case, the effects of $\mathcal{E}_A$ live in the pointed polyhedral cone with origin ${\mathbf 0}$ and extremal rays $\{\delta_{a|x}\}_{a,x}$, where $\delta_{a|x}$ denotes the standard basis vectors with entry $1$ at position $p(a|x)$. The fine-grained measurements are in this case \beq
            M = \{c_x\delta_{a|x}\}_{x,a},\qquad c_x
              \geq 0, \quad\sum_x c_x
              = 1, \label{eq:gw-fiducial-measurements-main}
        \eeq
(up to proportional splitting of effects) as we show in Appendix~\ref{app:ghostworld_GPT}.
The bipartite boxworld state space $\Omega_{AB}^{BW}$ is in one-to-one correspondence with the set of all non-signalling correlations, i.e, behaviors $\omega_{AB}=\{p(a,b|x,y)\}$ such that 
\begin{equation}
\begin{aligned}
  \sum_a p(a,b|x,y) &=\sum_{a} p(a,b|x',y) \quad \forall\, x,x',b,y,  \\
  \sum_b p(a,b|x,y) &=\sum_{b} p(a,b|x,y') \quad \forall\, y,y',a,x .
\end{aligned}  
\end{equation}
The bipartite effect space is the maximal possible one, $\mathcal{E}_{AB}^{BW}:= \{e \mid e(\omega_{AB}) \in [0,1] \ \forall\, \omega_{AB} \in \Omega_{AB}^{BW} \}$.

The ghostworld GPTs inherit the elementary local state spaces (as well as effect
spaces) from boxworld. The bipartite state space of the
$n$-ghostworld theory is $\Omega_{AB}^{(n)} \equiv \GW_n$, with joint effects
and measurements inherited from boxworld:
\beq
    \mathcal E_{AB}^{(n)} = \left\{ e\big|_{\Omega_{AB}^{(n)}} \colon e \in \mathcal E_{AB}^{\mathrm{BW}} \right\},
\eeq
where $|_{\Omega_{AB}^{(n)}}$ denotes a restriction of the domain. 
For any $m,d$ the sets $\GW_n^{(m,d)}$ are closed under wirings~\cite{lang:2014}. This implies that the $n$-ghostworld GPT introduced here leads to the set of non-local correlations $\GW_n^{(m,d)}$ for any $m,d$. For more properties of $n$-ghostworld and the corresponding sets of correlations see Appendix~\ref{app:ghostworld_GPT}.

In the IC proof, we regard the extensions $Q_A,Q_B$ of Definition~\ref{def:gw}
as auxiliary boxworld states and evaluate their measurement entropies in that
theory. 
Their copied subsystems, called \emph{ghosts}, are not available to the
parties in the physical protocol. More details and proofs are provided in
Appendix~\ref{app:ghostworld_GPT}. 

The \emph{measurement entropy} on a state $\omega_A$ is the infimum of the
Shannon entropy of the outcomes of a fine-grained
measurement~\cite{short:2010a}:
\beq
    H_+(A) = \inf_{\{e_k\} \in \mathcal M^*} H\bigl(\{e_k(\omega_A)\}_k\bigr). \label{eq:bg_measurement_entropy}
\eeq
It reduces to Shannon and von Neumann entropies for classical and quantum
systems, respectively. For a bipartite state $\omega_{AB}$, we use the
conditional measurement entropy~\cite{short:2010a, weilenmann:2020}
\beq
    H_+(A \mid B) = \inf_{\{f_k\} \in \mathcal M_B} \sum_k p(k)H_+(A|_{k}), \label{eq:bg_conditional_measurement_entropy}
\eeq
where $p(k) = f_k(\omega_B)$ and
$\omega_{A|_k} = (\mathrm{id}_A \otimes f_k)(\omega_{AB})/p(k)$ for
$p(k) > 0$.
This conditional entropy is always nonnegative and generally differs from the 
conditional von Neumann entropy for quantum states, but it reduces to the
conditional Shannon entropy in the classical case. Finally, we define the
generally asymmetric quantity
\beq
    \label{eq:bg_measurement_information}
    I_+(A \to B) \coloneqq H_+(A) - H_+(A \mid B),
\eeq
in analogy with the mutual information.

Since the ghostworld GPTs inherit the elementary local state and effect spaces from boxworld, their single system measurement entropy for $\omega_B=\{p(b|y)\}$ simplifies to
\beq
    \label{eq:bg_local_box_entropy}
    H_+(B) = \min_y H(B_y) = \min_y H\left( p(\cdot|y) \right),
\eeq
where $B_y$ denotes the outcome of measurement $y$. For classical states $C,M$,
the optimization defining $H_+(C \mid B,M)$ allows the measurement on $B$ to
depend on the observed value of $M$, so Bob can choose his measurement after
reading the message. In Appendix~\ref{app:ghostworld_GPT}, we prove
Eq.~\eqref{eq:bg_local_box_entropy} and derive similar simplifications for conditional measurement entropies of ghostworld systems.

\subsection{Contraction inequalities}

The technical tools used in our work include contraction coefficients of
divergences and the associated strong data-processing inequalities, widely studied concepts in classical~\cite{anantharam:2013, anantharam:2013b, raginsky:2016} and quantum~\cite{lesniewski:1999, hirche:2022} information theory.
A related connection between non-local correlations and hypercontractivity was established by Beigi and Gohari~\cite{beigi:2015}, who showed that
hypercontractivity ribbons are monotones under wirings.

For
probability distributions $q,p$ on the same finite alphabet, let
\beq
    \DKL(q\|p) = \sum_a q(a)\log\frac{q(a)}{p(a)}
\eeq
denote the Kullback--Leibler divergence. If $W$ is a classical channel
(represented by a column-stochastic matrix), we write $Wp$ for the output
distribution, with $(Wp)(b)=\sum_a W(b|a)p(a)$. The data-processing inequality for the KL divergence
reads $\DKL(W q\|W p) \leq \DKL(q\|p)$. For a fixed reference input distribution
$p$, the \emph{fixed-input contraction coefficient} of $W$ is defined
by~\cite{anantharam:2013, raginsky:2016, polyanskiy:2025}
\beq
    \eta(W,p) \coloneqq \sup_q \frac{\DKL(Wq\|Wp)}{\DKL(q\|p)}, \label{eq:bg_kl_contraction}
\eeq
where the supremum is over $q$ satisfying $0 < \DKL(q\|p) < \infty$. Thus
$0 \leq \eta(W,p) \leq 1$, and we have
\beq
    \label{eq:bg_relative_entropy_bound}
    \DKL(Wq\|Wp) \leq \eta(W,p) \DKL(q\|p)
\eeq
for all $q$ such that $\DKL(q\|p) < \infty$. A coefficient strictly below one
gives a \emph{strong data-processing inequality}.
Moreover, for any Markov chain $U \to A \to B$ with input distribution $P_A = p$
and channel $P_{B|A} = W$, it holds
\beq
    \label{eq:bg_mutual_information_contraction}
    I(U:B) \leq \eta(W,p)I(U:A).
\eeq
This follows by applying Eq.~\eqref{eq:bg_relative_entropy_bound} to each
conditional distribution $P_{A|U=u}$ and averaging over $u$. As we will show in the proofs of Results~\ref{res:contraction} and~\ref{res:contraction-extension}, contraction
coefficients can be used to bound the mutual information terms
appearing in IC, and thereby derive sufficient conditions for 
IC to hold.

\subsection{Reduction to deterministic protocols}
\label{sec:deterministic-reduction}

Alice and Bob's strategy for the IC task is specified
by a protocol, that we denote by $\pi$, specifying the (deterministic 
or probabilistic) functions (or maps) they apply to process their 
variables at each stage. Here we recall that to check whether 
Eq.~\eqref{eq:ic_inequality_noisy} is satisfied by a set of nonlocal boxes, one may restrict to a simpler class of \emph{deterministic 
protocols}~\cite{marco:costantino}.

For a single-copy protocol $\pi_1$, fix the behavior $P$ and $p(\rvec)$, and
consider the IC deficit 
${\Delta(P,\pi_1) = \sum_{s=0}^{n-1}I(R_s:G_s)
                                - I(\Rvec,X, A:M')}$.
With the communication channel $p(m'|\rvec,x,a)$ and all other
maps fixed, the joint distribution of all variables is affine in 
each local map. Because mutual information $I(X:Y)$ is 
convex in $p(y|x)$ for fixed $p(x)$ and concave in $p(x)$ for fixed 
$p(y|x)$, one has that $\Delta$ is convex in each local map separately, 
and its maximum can be attained with deterministic measurement choices
and decoding~\cite{marco:costantino}:
\beq
        x = f_X(\rvec),\quad
        y = f_Y(s,m'), \quad
        g = f_G(s,m',b),
\eeq
while the channel remains in general noisy. Since $X$ is then determined by $\Rvec$, the 
right-hand side of IC reduces to $I(\Rvec,A:M')$.

For a $c$-copy protocol $\pi_c$, Alice can choose which unused
copy to measure next and its input adaptively, based on $\rvec$ and her 
previous local choices and outcomes. Bob's choices may depend on
$s$, $m'$, and his previous local choices and outcomes, 
similarly to the characterization of general wirings~\cite{beigi:2015}.
Applying the same deterministic reduction to each choice of such maps~\cite{marco:costantino}, we may write
  \beq
      \begin{aligned}
          \ell_i
          &= f_{L_i}(\rvec,a_{\ell_1},\ldots,a_{\ell_{i-1}}), \\
          x_{\ell_i}
          &= f_{X_i}(\rvec,a_{\ell_1},\ldots,a_{\ell_{i-1}}), \\
          j_k
          &= f_{J_k}(s,m',b_{j_1},\ldots,b_{j_{k-1}}), \\
          y_{j_k}
          &= f_{Y_k}(s,m',b_{j_1},\ldots,b_{j_{k-1}}).
          \label{eq:deterministic}
      \end{aligned}
  \eeq
  Here, Alice measures the $\ell_i$th box at her time-step $i$ with input $x_{\ell_i}$, 
  and Bob measures his $j_k$th box at his time-step $k$ with input $y_{j_k}$. In particular,
  Alice's and Bob's orders might differ.
  Bob's final guess is then $g=f_G(s,m',\bvec)$. 
  For these deterministic protocols, Alice's order and inputs are determined by $\Rvec,\mathbf A$, therefore to check IC with uncorrelated inputs it suffices to consider
\beq
\sum_{s=0}^{n-1}I(R_s:G_s) \leq I(\Rvec,\mathbf A : M').
\eeq

\section{Insufficiency of finite-length information causality}
\label{sec:insufficiency-finite-length}

The IC inequality in Eq.~\eqref{eq:ic_inequality_noisy} is one instance of many
formulations of the principle that have been considered. In particular, it
extends the noisy-channel formulation~\cite{miklin:2021, jain:2024} to an
arbitrary number of copies, and subsumes the original
definition~\cite{pawlowski:2009}. Other versions have been proposed, namely IC
with correlated inputs~\cite{jain:2024, chaves:2015c}, relaxing the independence
assumption between the database entries,
\beq
    \label{eq:technical-part-ic-with-correlated-entries}
    \sum_{s=0}^{n-1}I(R_s:G_s \mid \Rvec_{<s}) \leq I(\Rvec,\mathbf A:M'),
\eeq
where $\Rvec_{<s} = (R_0, \dots, R_{s-1})$. 
Notice that the deterministic reduction 
applies to this expression as well, since $p(\rvec)$ is fixed, and each term on the left-hand side is an average of $I(R_s:G_s \mid \Rvec_{<s}=\rvec_{<s})$ to which the same convexity argument applies. This inequality was shown to be more
constraining under specific protocols, and later extended to general distributed
computation tasks~\cite{miklin:2026}. In the rest of the paper, we state our
results in terms of this inequality, and by ``the IC inequality'' we will refer
to Eq.~\eqref{eq:technical-part-ic-with-correlated-entries}, which implies the usual IC inequality when the entries are independent.

In this section, we prove Result~\ref{res:gw-ic}, using properties of the
ghostworld GPT and conditional measurement entropies introduced in Section~\ref{sec:prliminary_gw}. Let us fix $P \in \GW_n$
in an arbitrary Bell scenario. We consider a database with $n + 1$
possibly correlated entries $\Rvec = (R_0,\ldots,R_n)$, and an independent
question $S \in [n + 1]$ for Bob.

The distributions $Q_A$ and $Q_B$ that extend $P$ are non-signaling across all
bipartitions, hence each of the two can be interpreted as a state of a
multipartite boxworld system~\cite{barrett:2007}. This interpretation is useful
to calculate measurement entropies involving Bob's ghost systems
$\Bghost = (B_0, \dots, B_n)$. %
However, note that these are just auxiliary
distributions, since physically Bob has access to his copy $B = B_0$, i.e., his
share of the bipartite ghostworld state $\omega_{AB} \in \Omega_{AB}^{(n)}$.
With this premise, we describe the IC protocol using the GPT formalism.

In the following, we will rely on the reduction to deterministic protocols from Section~\ref{sec:deterministic-reduction}. In the GPT formalism, after Alice's measurement $x$
with outcome $a$, the conditional boxworld state of Bob and his ghosts is
\beq
    \omega_{\Bghost}^{a|x} = \left\{ \frac{q_B(a,\bvec|x,\mathbf y)}{p_A(a|x)} \right\}_{\bvec,\mathbf y}. \label{eq:gw-conditional-extension}
\eeq
In particular, by non-signaling
$\sum_a p_A(a|x)\omega_{\Bghost}^{a|x} = \omega_{\Bghost}$ is
independent of $x$. For the deterministic protocols considered here,
$p(x|\rvec) = \delta_{x,x_\rvec}$. After Alice's measurement and communication,
discarding $X$ gives the reduced state
\beq
    \begin{aligned}
        \omega_{\Rvec AM'\Bghost} & = \sum_{\mathbf r,a,m'}p(\mathbf r)p_A(a|x_\rvec)                                                  \\
                                       & \qquad\qquad \times p(m'|\mathbf r,a)\,[\mathbf r,a,m'] \otimes \omega_{\Bghost}^{a|x_\rvec}.
    \end{aligned} \label{eq:gw-post-channel-state}
\eeq
where $[c]$ denotes a pure classical state. Note that discarding $A,M'$ leaves
$\omega_\Rvec \otimes \omega_{\Bghost}$. Conditioning on $M'$ can
nevertheless correlate $\Rvec$ with $\Bghost$.

Measuring the ghost systems in this reduced state with inputs $\mathbf y$ gives
the joint distribution
\beq
    \begin{aligned}
        p(\rvec,a,m',\bvec|\mathbf y) = p(\rvec)p(m'|\rvec,a) q_B(a,\bvec|x_\rvec,\mathbf y).
    \end{aligned} \label{eq:gw-main-post-channel}
\eeq
All measurement entropies below are evaluated on the state in
Eq.~\eqref{eq:gw-post-channel-state} and its marginals. In particular, we
consider measurements on $M'\Bghost$. Since $M'$ is classical, we can first
read its value and then choose a joint measurement
$\{f_{k|m'}\}_k$ on $\Bghost$~\cite[Appendix~A]{weilenmann:2020}, with
outcome $k$, and write
\beq
    H_+(\Rvec \mid M',\Bghost) = \inf_{\{f_{k|m'}\}_k \in \mathcal M_{\Bghost}} H(\Rvec \mid M',K). \label{eq:gw-classical-target-measurement}
\eeq
We prove that in this setting two entropy bounds can be derived that together
give the IC inequality.

\begin{lemma}
    \label{lem:ghost-information-budget}
    For any deterministic protocol, considering the state after the message is
    received yields
    \beq
        \label{eq:ghost-information-budget}
        I_+(\Rvec \to M',\Bghost) \leq I(\Rvec,A:M').
    \eeq
\end{lemma}

\begin{proof}
 Let us first fix a message $M'=m'$ with $p(m')>0$ and a family $\{f_{k|m'}\}_k$ of measurements on $\Bghost$, chosen according
    to the received message. We then show that
 \begin{equation}
     H(\Rvec \mid M' = m',K) \geq H(\Rvec,A \mid M' = m')- h_{m'}, 
 \end{equation}
 where $h_{m'}=- \sum_{\mathbf r,a} p(\rvec,a|m') \log(p(a|\rvec))$.
 To see this, notice that 
 \begin{equation}
 \begin{aligned}
H&(\Rvec,A \mid M' = m')- h_{m'} - H(\Rvec \mid M' = m',K) \\
&= \sum_{\rvec,a,k}p(\rvec,a,k|m') \log\frac{p(\rvec|m',k)p(a|\rvec)}{p(\rvec,a|m')} \\
&\leq \log \left(\sum_{\rvec,a,k}p(\rvec,a,k|m') \frac{p(\rvec|m',k)p(a|\rvec)}{p(\rvec,a|m')} \right), \label{eq:negativity}
\end{aligned}
 \end{equation}
 where we use the concavity of the logarithm. To see that this is non-positive, notice that 
 \begin{equation}
     \begin{aligned}
     \sum_{\rvec,a,k}&p(\rvec,a,k|m') \frac{p(\rvec|m',k)p(a|\rvec)}{p(\rvec,a|m')} \\
       &\leq  \sum_{\rvec,a,k}p(k|\rvec,a,m') p(\rvec|m',k)p(a|\rvec)\\ 
         &= \sum_{\rvec,k} p(\rvec|m',k) \sum_a f_{k|m'}(\omega_{\Bghost}^{a|x_\rvec}) p(a|\rvec) \\
         &= \sum_{\rvec,k} f_{k|m'}(\omega_{\Bghost}) p(\rvec|m',k) \\
         &= 1
     \end{aligned}
 \end{equation}
 where the inequality follows because $p(k|\rvec, a,m') >0$ is possible also in cases where $p(\rvec, a |m') =0$; in the second equality we use linearity of the effects and non-signalling. If $p(\rvec, a |m')=0$, 
 we can assign $p(k|\mathbf r,a,m') = f_{k|m'}(\omega_{\Bghost}^{a|x_\rvec})$, and if $p(k|m')=0$,
 we can define $p(\rvec|m',k)$ as any normalized distribution.
 
    Averaging over $m'$ and then taking the infimum over $\{f_{k|m'} \}_k$, we obtain
    \beq
        \begin{aligned}
            H_+(\Rvec \mid M',\Bghost) & \geq H(\Rvec,A \mid M') - H(A \mid \Rvec) \\
                                            & = H(\Rvec) - I(\Rvec,A:M'),
        \end{aligned} \label{eq:gw-adaptive-measurement-entropy}
    \eeq
    which is the claimed bound.
\end{proof}

The second needed inequality is supplied by the $n$-ghostworld assumption. The
idea is that for each received message, one can measure each ghost $B_i$ with an
input $y_i$ that minimizes the remaining uncertainty about $R_i$. The resulting
product measurement is an allowed joint measurement on the extension, so this implies an inequality resembling a 
subadditivity, Eq.~\eqref{eq:ghost-subadditivity}.

\begin{lemma}
    \label{lem:ghost-entropy}
    For the state after the received message,
    \begin{equation}
        \begin{aligned}
            H_+(\Rvec \mid M',\Bghost) \leq \sum_{i=0}^{n}H_+(R_i \mid M',B). \label{eq:ghost-subadditivity}
        \end{aligned}
    \end{equation}
    Moreover,
    \beq
        H_+(\Rvec \mid M',\Bghost) \leq \sum_{i=0}^{n}H(R_i \mid \Rvec_{<i},M',G_i), \label{eq:ghost-correlated-entropy}
    \eeq
\end{lemma}
\begin{proof}
    We restrict to measurements that read $M'$ and then perform one
    measurement on each ghost, with inputs $\mathbf y$ chosen according to $m'$.
    These product measurements are allowed, so they give an upper bound on the
    conditional measurement entropy in
    Eq.~\eqref{eq:gw-classical-target-measurement}:
    \beq
        \begin{aligned}
            & H_+(\Rvec \mid M',\Bghost)                                                                                     \\
            & \quad \leq \sum_{m'}\min_{\mathbf y}\sum_{\bvec} p(m',\bvec|\mathbf y)H(\Rvec|m',\bvec,\mathbf y)              \\
            & \quad \leq \sum_{m'}\min_{\mathbf y}\sum_{i=0}^{n}\sum_{\bvec} p(m',\bvec|\mathbf y)H(R_i|m',\bvec,\mathbf y). \\
            & \quad \leq \sum_{m'}\min_{\mathbf y}\sum_{i=0}^{n}\sum_{b_i} p(m',b_i|y_i)H(R_i|m',b_i,y_i)                    \\
            & \quad = \sum_{i=0}^{n}H_+(R_i \mid M',B_i).                                                                    \\
        \end{aligned} \label{eq:gw-ghost-proof}
    \eeq
    The second inequality uses Shannon subadditivity. The third uses concavity:
    for each $i$, averaging $p(r_i|m',\bvec,\mathbf y)$ over the other ghost
    outcomes $\bvec_{-i}$ with weights $p(\bvec_{-i}|m',b_i,\mathbf y)$ gives
    $p(r_i|m',b_i,\mathbf y)$, which equals $p(r_i|m',b_i,y_i)$ by
    non-signaling. The entropy of this mixture is at least the corresponding
    average of the entropies. The last step follows by minimizing each $y_i$
    independently, and using a measurement entropy identity (see
    Eq.~\eqref{eq:gw-classical-conditional-entropy} in
    Appendix~\ref{app:ghostworld_GPT}). Because the ghosts are identical, we get
    the bound of Eq.~\eqref{eq:ghost-subadditivity}.

    To prove the second part of the lemma, use the protocol input
    $y_i = f_Y(i,m')$ and decoding $g_i = f_G(i,m',b_i)$ on the $i$th ghost. The
    resulting joint guesses $\mathbf G = (G_0,\ldots,G_n)$ have the same
    individual marginals $(\Rvec,M',G_i)$ as the physical protocol. For each
    $m'$, the joint guesses can be regarded as the outcome of an allowed
    measurement on $\Bghost$, obtained by classical processing of the
    product-measurement outcomes, therefore like before we get an upper bound
    for the measurement entropy:
    \beq
        \begin{aligned}
            H_+(\Rvec \mid M',\Bghost) %
            & \leq \sum_{m',\mathbf g}p(m',\mathbf g)H_+(\Rvec|m',\mathbf g) \\
            & = H(\Rvec \mid M',\mathbf G) \\
            & = \sum_{i=0}^{n}H(R_i \mid \Rvec_{<i},M',\mathbf G) \\
            & \leq \sum_{i=0}^{n}H(R_i \mid \Rvec_{<i},M',G_i),
        \end{aligned} \label{eq:gw-correlated-ghost-proof}
    \eeq
    We have used that the conditional states of $\Rvec$ are classical, so their
    measurement entropies are Shannon entropies. The last two steps are the
    Shannon entropy chain rule and data processing, discarding all guesses
    except $G_i$ in each summand.
\end{proof}

\begin{theorem}
    \label{thm:gw-ic}
    For all $n \geq 1$, any correlation in $\GW_n$ satisfies information
    causality for $n + 1$ correlated database entries:
    \beq
        \label{eq:gw-noisy-ic}
        \sum_{s=0}^{n}I(R_s:G_s \mid \Rvec_{<s}) \leq I(\Rvec,\mathbf A:M').
    \eeq
    The statement holds for any finite number $c$ of independent $\GW_n$
    behaviours. In particular, $\GW_n \subseteq \IC_{n+1}$.
\end{theorem}
\begin{proof}
    Start by considering one copy only of $P$, and write
    $\mathcal I \coloneqq \sum_{i=0}^{n}I(R_i:G_i \mid \Rvec_{<i})$ for the
    left-hand side of IC. Use the Shannon entropy chain rule
    $H(\Rvec) = \sum_{i=0}^{n}H(R_i \mid \Rvec_{<i})$ and the two lemmas:
    \beq
        \label{eq:gw-ic-proof-chain}
        \begin{aligned}
            \mathcal I & \leq H(\Rvec) - \sum_{i=0}^{n}H(R_i \mid \Rvec_{<i},M',G_i) \\
                       & \leq I_+(\Rvec \to M',\Bghost)                         \\
                       & \leq I(\Rvec,A:M').
        \end{aligned}
    \eeq
    In the first inequality, we add $M'$ to the conditioning in each entropy.
    The second uses Lemma~\ref{lem:ghost-entropy}, and the third is
    Lemma~\ref{lem:ghost-information-budget}. 
    Finally, consider a protocol using independent behaviors $P_1,\ldots,P_c \in \GW_n$. This protocol defines an    effective behavior with external inputs $\rvec$ for Alice and $(s,m')$ for    Bob, and outputs $\mathbf a$ and $\bvec$~\cite{marco:costantino}. Since $\GW_n$ is    closed under wirings~\cite{lang:2014}, this behavior belongs to $\GW_n$.    Using it with the same channel $p(m'|\rvec,\mathbf a)$ and decoding strategy    reproduces the original protocol. Applying the single-copy argument with $A$    replaced by $\mathbf A$ proves Eq.~\eqref{eq:gw-noisy-ic} for every finite    $c$. 
\end{proof}

We now construct the nonlocal non-quantum correlations mentioned in Result~\ref{res:gw-ic}.
Let us consider the $(m,m,2,2)$ scenario, with $x,y \in \{1,\ldots,m\}$, $a,b \in \{0,1\}$, and set
$p_x^A = p_A(0|x)$, $p_y^B = p_B(0|y)$, and $p_{xy} = p(0,0|x,y)$ for the Collins--Gisin
representation~\cite{collins:2004a}. Consider the functional
\beq
    I_{mm22}(P) = \sum_{x,y=1}^{m}M_{xy}p_{xy}
                  - \sum_{y=1}^{m}(m - y)p_y^B
                  - p_1^A, \label{eq:gw-bell-functional}
\eeq
where
\beq
    M_{xy} = \begin{cases}
                 1,  & x + y \leq m + 1, \\
                 -1, & x + y = m + 2,    \\
                 0,  & x + y \geq m + 3.
             \end{cases} \label{eq:gw-bell-coefficients}
\eeq
Its local bound is zero~\cite{collins:2004a}.

\begin{proposition}
    \label{prop:bell-upper}
    For every $m \geq 2$, the behavior $P_m$ in the $(m,m,2,2)$ Bell scenario
    with the following entries
    \begin{equation}
        \begin{behaviorarray}{c|cccc}
            xy \ \backslash \ ab & 00  & 01      & 10  & 11  \\
            \hline
            x+y=m+2              & 0   & (m-1)/m & 1/m & 0   \\
            x+y\ne m+2           & 1/m & (m-2)/m & 0   & 1/m
        \end{behaviorarray} \label{eq:bell-full-table}
    \end{equation}
    belongs to $\GW_{m-2}$ and attains $I_{mm22}(P_m) = 1/m$.
\end{proposition}
\begin{proof}
Equation~\eqref{eq:bell-full-table} defines a normalized non-signaling
    behavior, with $p_x^A = (m - 1)/m$ and $p_y^B = 1/m$. Since $p_{xy} = 0$
    whenever $M_{xy} = -1$, the $m + 1 - y$ positive joint terms in each column
    $y$ give
    \beq
        \begin{aligned}
            I_{mm22} & = \frac1m\sum_{y=1}^{m}\bigl[(m + 1 - y) - (m - y)\bigr] - \frac{m - 1}{m} \\
                    & = \frac1m. \label{eq:gw-pm-value}
        \end{aligned}
    \eeq
    To prove that $P_m \in \GW_{m-2}$, we construct the needed extensions. 

Introduce the set of labels $\mathcal D = \{1,\ldots,m+1\}$.
Assign label $x$ to Alice's input $x$ and label $m+2-y$ to
Bob's input $y$. The two labels coincide precisely when
$x+y=m+2$, which, by the definition of $P_m$, is equivalent
to $p(0,0\mid x,y)=0$.
For $\lambda \in \mathcal D$, let us define
        \beq
            a_\lambda(x) = \begin{cases}
                               1 & \lambda = x   \\
                               0 & \lambda \ne x
                           \end{cases}, \quad b_\lambda(y)
                         = \begin{cases}
                               0 & \lambda = m + 2 - y   \\
                               1 & \lambda \ne m + 2 - y
                           \end{cases}. \label{eq:gw-label-responses}
        \eeq
        Each required extension has $m$ parties: $m - 1$ copies of one side and
        one of the other. Its inputs therefore specify at most $m$ distinct
        labels. Omit one label absent from the input tuple and sample uniformly
        from the other $m$ labels, using the responses above.

        Explicitly, set $D_A = \{x_0,\ldots,x_{m-2},m + 2 - y\}$ and
        $D_B = \{x,m + 2 - y_0,\ldots,m + 2 - y_{m-2}\}$. Choose
        $\ell_A \in \mathcal D \setminus D_A$ and
        $\ell_B \in \mathcal D \setminus D_B$, and define
        $\Lambda_A = \mathcal D \setminus \{\ell_A\}$ and
        $\Lambda_B = \mathcal D \setminus \{\ell_B\}$. Then
        \beq
            \begin{aligned}
                q_A(\mathbf a,b|\mathbf x,y) & = \frac1m\sum_{\lambda \in \Lambda_A} \left( \prod_{i=0}^{m-2}\delta_{a_i,a_\lambda(x_i)} \right) \delta_{b,b_\lambda(y)}, \\
                q_B(a,\bvec|x,\mathbf y)     & = \frac1m\sum_{\lambda \in \Lambda_B} \delta_{a,a_\lambda(x)} \left( \prod_{j=0}^{m-2}\delta_{b_j,b_\lambda(y_j)} \right).
            \end{aligned} \label{eq:gw-pm-extensions}
        \eeq
        Both distributions are nonnegative and normalized. One can also check
        that both extensions are non-signaling across every bipartition. They
        are also independent of the omitted label and symmetric under
        permutations of copies.

        Finally, we verify that any Alice-Bob marginal of either extension
        equals $P_m$. If $x = m + 2 - y$, the event $\lambda = x$ has
        probability $1/m$ and gives $(a,b) = (1,0)$, while every other sampled
        $\lambda$ gives $(0,1)$. If $x \ne m + 2 - y$, sampling $x$ gives
        $(1,1)$ and sampling $m + 2 - y$ gives $(0,0)$, each with probability
        $1/m$; the remaining $m - 2$ labels give $(0,1)$. These are precisely
        the two rows of Eq.~\eqref{eq:bell-full-table}.
\end{proof}

The extension in the proof might look like a local hidden-variable model, in apparent
    contradiction with the nonlocality of $P_m$. However, the sets $\Lambda_A$,
    $\Lambda_B$, and hence the distribution of $\lambda$, depend on the inputs.
    Therefore, the constructions violate measurement independence, so there is
    no contradiction.
The value $1/m$ achieved by $P_m$ is optimal over $\GW_{m-2}$ (see
Appendix~\ref{app:ghostworld_GPT} for a proof). Furthermore, the fact that $p(0,0|x,y) = 0$ when $x + y = m + 2$ eliminates all terms
with coefficient $-1$ in $I_{mm22}$. Using this specific \emph{zero pattern}, we
can prove by contradiction that it lies outside the commuting quantum set.

\begin{proposition}
    \label{prop:bell-postquantum}
    For every $m \geq 2$, $P_m \notin \Qc$.
\end{proposition}
\begin{proof}
    By contradiction, suppose $P_m$ has a commuting-operator realization,
    without loss of generality with a pure state $|\psi\rangle$ and projective
    measurements $\{A_{a|x}\}$, $\{B_{b|y}\}$~\cite{navascues:2008}. 
    The probability table gives
    \beq
    \begin{aligned}
        \langle\psi|A_{1|x}|\psi\rangle
        = \langle\psi|B_{0|y}|\psi\rangle
        = \frac 1 m, \\
        \langle\psi|A_{1|x}B_{0|y}|\psi\rangle
        = \frac{\delta_{x,m+2-y}}m. \label{eq:gw-rare-events}
\end{aligned}
    \eeq
    Define $|u_x\rangle = A_{1|x}|\psi\rangle$ and $|v_y\rangle = B_{0|y}|\psi\rangle$.
    Since $A_{1|x}^2 = A_{1|x}$ and $B_{0|y}^2 = B_{0|y}$,
    \beq
        \|u_x\|^2 = \|v_y\|^2
                  = \frac 1 m, \qquad\langle u_x|v_y\rangle
                  = \frac{\delta_{x,m+2-y}}m.
    \eeq
    For $2 \leq y \leq m$, these identities imply
    $|v_y\rangle=|u_{m+2-y}\rangle$. Consequently, 
    for all $x \ne x'$ with $x'\ge 2$,
    \beq
    \begin{gathered}
    \langle u_x | u_{x'} \rangle = \langle u_x | v_{m+2-x'} \rangle = 0, \\ 
    \langle u_x | v_1 \rangle = \delta_{x,m+1}/m = 0,
    \end{gathered}
    \eeq
    i.e.,
    $|u_1\rangle,\ldots,|u_m\rangle,|v_1\rangle$ are pairwise 
    orthogonal. %
    Moreover each of them is a projection of the unit vector $|\psi\rangle$.
    Therefore,
    \beq
        \frac{m + 1}{m}
        = \sum_{x=1}^{m}\|u_x\|^2 + \|v_1\|^2
        \leq \|\psi\|^2 = 1,
    \eeq
    a contradiction. Thus $P_m \notin \Qc$.
\end{proof}

\section{Insufficiency of single-copy information causality}
\label{sec:insufficiency-single-copy}

In this section, we focus on IC with database of any length $n \in \mathbb{N}$, but where Alice and Bob are restricted to use only a single copy of a behavior $P=\{p(a,b|x,y)\}$. 
For this it will be useful to think of this behaviour as follows: define
\beq
    W_{xy}(b|a) = \frac{p(a,b|x,y)}{p_A(a|x)}, \label{eq:contraction-channels}
\eeq
assuming $p_A(a|x) \ne 0$ for all $a,x$. 
If we fix $x,y$, we can regard $W_{xy}$ as a classical channel $A \to B$, because for all $a,b$, $W_{xy}(b|a) \ge 0$ and $\sum_{b} W_{xy}(b|a) = 1$. Since all behaviors we consider are non-signalling, we can write 
$p_B(b|y) = \sum_a W_{xy}(b|a) p_A(a|x)$, or for short
\beq
    \label{eq:W_matrix_vector_product}
    \bm p_{B|y} = W_{xy} \bm p_{A|x}.
\eeq

Intuitively, when a behaviour is used in an IC protocol, the local processing of inputs and outcomes by Alice and Bob cannot change the channel $W_{xy}$ but only the distributions of their inputs $x,y$ (also dependent on the sent message $M'$). Thus if for each $x,y$ the channel is sufficiently contractive, Alice's change in $x$ (conditioned on $\Rvec$ and $M'$) cannot affect Bob's outcome distribution significantly enough for the behavior to be useful for the IC task. To make this intuition precise, we first prove the following theorem.

\begin{theorem}[Contractivity criterion for information causality]
    \label{thm:contraction-criterion}
    Let Alice and Bob share a non-signalling behavior $P=\{p(a,b|x,y)\}$ that satisfies
    $p_A(a|x) > 0$ for all $a,x$, and such that for every input pair $(x,y)$
    there exist non-negative constants $\eta_{xy}$ such that
    \beq
        \sum_y \eta_{xy} \leq 1 \qquad \text{for every }x \label{eq:contraction-row-sums}
    \eeq
    and for every distribution $\bm w = w(\cdot)$ of Alice's outcomes,
    \beq
        \label{eq:relative-entropy-criterion}
        \DKL\left( W_{xy} \bm w \, \middle\| \, W_{xy} \bm p_{A|x} \right) \leq \eta_{xy}\DKL\left( \bm w \,\middle\|\, \bm p_{A|x} \right).
    \eeq
    Then, for every deterministic single-copy protocol $\pi$ and for every
    database length $n$ and arbitrarily correlated entries, the following
    inequality holds:
    \beq
        \label{eq:middle-ic-inequality-operational}
        I(\Rvec : M') + \sum_{y} I(\Rvec :B_y \mid M') \leq I(\Rvec,A:M').
    \eeq
    In turn, this implies the correlated-input IC inequality,
    \beq
        \sum_{s=0}^{n-1}I(R_s:G_s \mid \Rvec_{<s}) \le I(\Rvec,A:M') \label{eq:contraction-correlated-ic}
    \eeq
    for any $n$. 
\end{theorem}

Recalling the reduction to deterministic protocols from Section~\ref{sec:deterministic-reduction}, we immediately obtain the following corollary.

\begin{corollary}
    Any behavior $P$ satisfying the conditions of Theorem~\ref{thm:contraction-criterion} satisfies single-copy information causality (for any single-copy protocol).
\end{corollary}

The smallest admissible constants in Eq.~\eqref{eq:relative-entropy-criterion}
are the fixed-input contraction coefficients $\eta(W_{xy},\bm p_{A|x})$ defined
in Eq.~\eqref{eq:bg_kl_contraction}. The proof of
Theorem~\ref{thm:contraction-criterion} is given in
Appendix~\ref{app:contraction}.

To apply the criterion, let us consider the families of behaviors
\beq
    P_\lambda
    = \begin{behaviorarray}{c|cccc}
          xy \ \backslash \ ab & 00         & 01       & 10         & 11         \\
          \hline
          00,01,10,11     & q          & 0        & 1-2q       & q          \\
          02              & 0          & q        & \tfrac12   & \tfrac12-q \\
          12              & q          & 0        & \tfrac12-q & \tfrac12   \\
          20              & \tfrac12-q & q        & \tfrac12   & 0          \\
          21              & \tfrac12   & 0        & \tfrac12-q & q          \\
          22              & \tfrac14   & \tfrac14 & \tfrac14   & \tfrac14
      \end{behaviorarray}. \label{eq:contraction-segment-table}
\eeq
for $0 \leq \lambda \leq 1$ and where 
$q = (3 + \lambda)/12$. It is straightforward to verify that $P_\lambda \in \GW_1^{(3,2)}$ for all $\lambda \in [0,1]$, by considering the decomposition
\beq
    P_\lambda = (1 - \lambda)P_L + \lambda P_G, \label{eq:contraction-segment}
\eeq
where 
    \beq
        \begin{aligned}
                    P_G &= \begin{behaviorarray}{c|cccc}
                                                          xy \ \backslash \ ab & 00       & 01       & 10       & 11       \\
                                                          \hline
                                                          00,01,10,11     & \tfrac13 & 0        & \tfrac13 & \tfrac13 \\
                                                          02              & 0        & \tfrac13 & \tfrac12 & \tfrac16 \\
                                                          12              & \tfrac13 & 0        & \tfrac16 & \tfrac12 \\
                                                          20              & \tfrac16 & \tfrac13 & \tfrac12 & 0        \\
                                                          21              & \tfrac12 & 0        & \tfrac16 & \tfrac13 \\
                                                          22              & \tfrac14 & \tfrac14 & \tfrac14 & \tfrac14
                                                      \end{behaviorarray},
        \end{aligned} \label{eq:contraction-endpoint-tables}
    \eeq
        and
   \beq
        P_L = \frac14\left( P^{010;001} + P^{101;000} + P^{110;100} + P^{111;011} \right), \label{eq:contraction-local-decomposition}
    \eeq
    where $P^{a_0a_1a_2;b_0b_1b_2}$ denotes the
    deterministic local behavior in which Alice outputs $a_x$ on input $x$ and
    Bob outputs $b_y$ on input $y$. 
    While $P_L \in \mathcal L ^{(3,2)} \subsetneq \GW_1^{(3,2)}$ as a convex combination of local deterministic distributions, it was shown in~\cite[Section~5]{lang:2014} that also $P_G \in \GW_1^{(3,2)}$. Thus convexity of $\GW_1^{(3,2)}$ proves the claim.

In the following we further show that the family of behaviours we have defined, as soon 
as $\lambda \ne 0$, does not intersect with the quantum set.

\begin{proposition}
    \label{prop:contraction-segment}
    For every $0 < \lambda \leq 1$, $P_\lambda$ from Eq.~\eqref{eq:contraction-segment-table} satisfies $P_\lambda \notin \Qc^{(3,2)}$. Furthermore, for
    \beq
        0 \leq \lambda
          \leq 21 - 12\sqrt3
          \approx 0.215, \label{eq:contraction-segment-threshold}
    \eeq
    it satisfies the hypotheses of Theorem~\ref{thm:contraction-criterion}.
\end{proposition}

To prove this, we observe that for the family under consideration, the channels
$W_{xy}$ are, up to outcome relabellings, binary $Z$-channels~\cite{polyanskiy:2025}, except
for $W_{22}$, whose output is independent of its input. We use the exact
fixed-input contraction coefficients for binary $Z$-channels derived in  \cite{anantharam:2013b} to verify the criterion throughout the
interval in Eq.~\eqref{eq:contraction-segment-threshold}.
 All the behaviors, calculations and the full proofs are given in
Appendix~\ref{app:contraction}.

Choosing for instance $\lambda = 1/5$ gives the explicit behavior that for single copy information causality satisfies
\beq
    P^* \coloneqq P_{1/5} 
        \in \IC_{\infty,1}^{(3,2)} \setminus \Qc^{(3,2)}, \label{eq:contraction-pstar}
\eeq
where the subscript $1$ indicates the single-copy restriction.
Moreover, single-copy information causality sets are convex: this 
follows directly from the convexity of the mutual information in the channel when the marginal held by Alice is fixed, and it has been previously established in the literature~\cite{yu:2026, marco:costantino}. %
Since all local points satisfy IC, this directly implies the following corollary and thus Result~\ref{res:contraction}, including its correlated-input
version. 

\begin{corollary}
The set 
    $\mathcal S^{(3,2)}= \conv( \mathcal{L}^{(3,2)} \cup P_{1/5})$ when restricted to using a single copy of the behavior, satisfies information causality. Furthermore, $P_{1/5} \notin \Qc^{(3,2)}$.
\end{corollary}

\section{Insufficiency of general information causality}
\label{sec:insufficiency-general}

The contraction criterion of Theorem~\ref{thm:contraction-criterion} is also
sufficient to certify membership to $\IC$ with arbitrarily many independent
copies. Here we state this extension.

\begin{theorem}[Multi-copy contraction criterion]
    \label{thm:contraction-multi-copy}
    Let $P$ satisfy the hypotheses of Theorem~\ref{thm:contraction-criterion}.
    Then, for every database length $n$ and every $c \geq 1$, every
    deterministic adaptive protocol using $c$ independent copies of $P$
    satisfies
    \beq
        \label{eq:contraction-multi-copy-ic}
        \sum_{s=0}^{n-1} I(R_s:G_s \mid \Rvec_{<s}) \leq I(\Rvec,\mathbf A:M'),
    \eeq
    where $\mathbf A = (A_1,\ldots,A_c)$ collects Alice's outcomes. The same
    holds with $(M',G_s)$ in place of $G_s$. In particular,
    $P \in \IC_\infty$.
\end{theorem}

 Consider $c$ independent
    boxes $P_1, \ldots, P_c$.
    We index all inputs and outcomes by these labels, writing
    $\mathbf x = (x_1,\ldots,x_c)$, $\mathbf y = (y_1,\ldots,y_c)$ and 
    $\mathbf a = (a_1,\ldots,a_c)$, $\bvec = (b_1,\ldots,b_c)$ respectively.
    To prove the theorem, we use the  reduction to
    deterministic protocols from Section~\ref{sec:deterministic-reduction}.
    
    For a fixed protocol $\pi_c$, of
    this kind, the joint distribution is
    \beq
    \begin{aligned}
        p&^{(P_1, \ldots ,P_c)}_{\pi_c}(\rvec,s,\mathbf a,m',\bvec,g) = \\
        &p(\rvec)p(s)\,p(m'|\rvec,\mathbf a)\delta_{g,f_G(s,m',\bvec)} \prod_{i=1}^{c}p_i(a_i,b_i|x_i,y_i),\label{eq:mc-mixed-protocol}
    \end{aligned}
    \eeq
    where the inputs $\mathbf x, \mathbf y$ are deterministic functions of other variables in the protocol (see Eq.~\eqref{eq:deterministic}).

    One of the main difficulties in the analysis is that the distributions of Bob's outcomes depend on the database via the behaviours $P_1, \ldots, P_c$ and the message combined. To decouple this we first consider the case where Bob's inputs are independent:
    For each copy $i$,  we introduce auxiliary variables $B_{i|y}$, one for every
    Bob input $y$, and write
   $ \Bvec_i
        \coloneqq (B_{i|y})_y$ and $\Bvec_{\mathrm{tot}}
        \coloneqq (\Bvec_1,\ldots,\Bvec_c)$.
    The idea is then to construct a distribution $Q(\rvec,m',\bvec_{\mathrm{tot}})$
    that reproduces Bob's guessing statistics for suitably chosen combinations of values from $\mathbf b_{\rm tot}$ and that satisfies
    $I_Q(\Rvec:M',\Bvec_{\mathrm{tot}}) \leq I(\Rvec,\mathbf A:M')$, where we use the subscript $Q$ to indicate that the entropy is taken on the auxiliary distribution $Q(\rvec,m',\bvec_{\mathrm{tot}})$. The idea is to then look at each individual behavior $P_1, \ldots, P_c$ as a channel from Alice to Bob, similarly to the single-copy case, and apply the contraction criterion at each copy. More specifically, for $P_1=\cdots = P_c = P$, we arrive at the following lemma.

    \begin{lemma}%
        \label{lem:mc-information-bound}
        Let $P$ satisfy the hypotheses of
        Theorem~\ref{thm:contraction-criterion}. For every database distribution
        and every deterministic adaptive protocol using $c \geq 1$ independent
        copies of $P$, there exists a distribution
        $Q(\rvec,m',\bvec_{\mathrm{tot}})$, independent of $s$, with the
        original marginal on $(\Rvec,M')$, such that
        \beq
            I_Q(\Rvec:M',\Bvec_{\mathrm{tot}}) \leq I(\Rvec,\mathbf A:M'). \label{eq:mc-auxiliary-information-bound}
        \eeq
        Moreover, for every $s$, there exists a deterministic function $F_s$
        such that $\widetilde G_s = F_s(M',\Bvec_{\mathrm{tot}})$ has the same
        joint distribution with $(\Rvec,M')$ as the original guess $G_s$.
    \end{lemma}

    The proof of this lemma is rather technical and thus deferred to Appendix~\ref{app:contraction-multi-copy}. Based on this, we can now prove the theorem.

        \begin{proof}[Proof of Theorem~\ref{thm:contraction-multi-copy}]
        Let $Q$ be the distribution of Lemma~\ref{lem:mc-information-bound}.
        Data processing and the lemma give, for arbitrarily correlated
        $\{R_s\}$,
        \beq
            \begin{aligned}
                \sum_{s=0}^{n-1}I(R_s:G_s|\Rvec_{<s}) & \leq \sum_{s=0}^{n-1} I_Q(R_s:M',\Bvec_{\mathrm{tot}}|\Rvec_{<s}) \\
                                                      & = I_Q(\Rvec:M',\Bvec_{\mathrm{tot}})                              \\
                                                      & \leq I(\Rvec,\mathbf A:M').
            \end{aligned} \label{eq:mc-correlated-ic}
        \eeq
        The first inequality uses conditional data processing and the defining
        property of $F_s$, whereas the equality uses the mutual information
        chain rule.
    \end{proof}

Combining this with Proposition~\ref{prop:contraction-segment} gives
Result~\ref{res:contraction-extension}, as for $0 < \lambda \leq 21 - 12\sqrt3$,
the behavior $P_\lambda$ of Eq.~\eqref{eq:contraction-segment-table} satisfies
\beq
    P_\lambda \in \IC_\infty^{(3,2)}
                  \setminus \Qc^{(3,2)}. \label{eq:contraction-multi-copy-segment}
\eeq
In combination with the following theorem, this establishes Result~\ref{res:contraction-extension}.

\begin{theorem}
    \label{thm:contraction-mixed}
    Let $P_1,\ldots,P_c \in \mathcal{S}^{(3,2)}$, where
    \beq
    \mathcal{S}^{(3,2)} \coloneqq \conv(\LHV^{(3,2)} \cup \{P^*\}).
    \eeq
    Then, for every database length $n$ and every adaptive
    protocol using the independent resources $P_1, \ldots, P_c$,
    information causality holds:
    \beq
        \sum_{s=0}^{n-1} I(R_s:G_s \mid \Rvec_{<s}) \leq I(\Rvec,\mathbf X, \mathbf A:M'). \label{eq:contraction-mixed-ic}
    \eeq
    The same is true with $(M',G_s)$ in place of $G_s$.
\end{theorem}

Notice that $\mathcal{S}^{(3,2)}$ is not closed under local relabellings, however, the conclusion of the theorem also holds for its closure under finite local wirings and shared randomness. Indeed, any IC protocol using independent resources from this closure can be turned into an IC protocol using finitely many independent resources from $\mathcal{S}^{(3,2)}$.

To prove the theorem, we proceed in two steps. First, we show by convexity arguments
that any IC violation using resources from $\mathcal S^{(3,2)}$ would also be possible
using extremal resources, in the same spirit as the deterministic reduction for protocols.
Second, we use the fact that protocols using copies of $P^*$ and local
deterministic boxes satisfy IC, since the local boxes can be simulated locally and 
copies of $P^*$ do not violate IC (by Theorem~\ref{thm:contraction-multi-copy}). 
Hence, Theorem~\ref{thm:contraction-mixed} rules out the possibility of reaching an IC 
violation by supplementing $P^*$ with local resources. 
It gives a convex set $\mathcal S$ containing all
local behaviors and some post-quantum ones, with IC holding for every adaptive
protocol combining independent resources from this set. The proofs are deferred
to Appendix~\ref{app:contraction-multi-copy}. 
Finally, using the Navascués--Pironio--Acín hierarchy~\cite{navascues:2008}
one can see that $P^*$ satisfies macroscopic locality~\cite{navascues:2010}, but
lies outside the almost quantum set~\cite{navascues:2015} (consistently with 
the observation that our proof that $P^* \notin \Q_c$ can be adapted to 
a proof of exclusion from almost quantum).
Our results therefore imply that 
there exist behaviors outside the almost quantum set that satisfy information causality.

\section{Conclusions}

We have shown that information causality does not single out the sets of bipartite quantum correlations, even in the case where an arbitrary number of behaviors can be used in the protocol,
already in the Bell scenario with three binary measurements per party.
This raises the immediate question whether IC can be further strengthened or, if not, whether there is an alternative principle that recovers bipartite quantum correlations without assuming the Hilbert space formalism. This would point us towards identifying what are the unique features of quantum theory.

A significant related problem is the relationship between information causality and the 
almost quantum set of correlations~\cite{navascues:2015}, a strict superset of quantum 
correlations that was shown to satisfy several other proposed physical principles~\cite{brassard:2006, linden:2007, navascues:2010, fritz:2013}. Our results show that some correlations that satisfy information causality lie outside the almost quantum set.
Whether all almost quantum correlations satisfy information causality remains an open question.

Our analysis also shows that IC can be established for GPTs (and corresponding sets of correlations) without defining an entropy that satisfies all the properties used in the usual GPT proof of IC~\cite{barnum:2010} (this is even possible from just the construction of bipartite state and effect spaces). In particular, the conditional measurement entropy we use does not reduce to the conditional von Neumann entropy in the quantum case. Exploring the consequences of potential other entropy measures for GPTs remains an interesting direction 
for future research.

\section*{Acknowledgements}
The authors thank Costantino Budroni, Jean-Daniel Bancal and Marco Túlio Quintino for insightful 
discussions.
This work was supported by the STeP2 grant (ANR-22-EXES-0013) of the Agence Nationale de la Recherche (ANR) as part of Plan France 2030.

\medskip

\section*{Statement on AI usage}

The conceptualization and development of the main ghostworld GPT and finite-length 
results (Result~\ref{res:gw-ic}) were performed by the authors. During the subsequent
study of the closure of information causality sets under wirings, OpenAI's
GPT-5.6 Sol was used for mathematical assistance. In conversations on the use of
contractivity in combination with our result,
the model suggested an initial version of
Theorem~\ref{thm:contraction-criterion} and pointed to relevant information
theory literature. It also provided an alternative proof that the ghostworld point
$P_G$ is not quantum, which was more suitable to be extended by the authors to the segment of
Proposition~\ref{prop:contraction-segment}. Building on the single-copy
criterion, the authors subsequently developed the multi-copy extension
establishing Result~\ref{res:contraction-extension}, using the model 
to check calculations and spot mathematical errors.
Finally, the model was used to correct grammatical mistakes in this manuscript. The authors have verified all arguments in the final manuscript and take full responsibility for its contents.

\bibliography{bibliography}{}

@article{al-safi:2011,
  title = {Information Causality from an Entropic and a Probabilistic Perspective},
  author = {{Al-Safi}, Sabri W. and Short, Anthony J.},
  year = 2011,
  month = oct,
  journal = {Physical Review A},
  volume = {84},
  number = {4},
  eprint = {1107.4031},
  primaryclass = {quant-ph},
  pages = {042323},
  issn = {1050-2947, 1094-1622},
  doi = {10.1103/PhysRevA.84.042323},
  urldate = {2026-10-05},
  archiveprefix = {arXiv}
}

@misc{anantharam:2013,
    title         = {On {{Maximal Correlation}}, {{Hypercontractivity}}, and the
                     {{Data Processing Inequality}} Studied by {{Erkip}} and
                     {{Cover}}},
    year          = 2013,
    author        = {Anantharam, Venkat and Gohari, Amin and Kamath, Sudeep and
                     Nair, Chandra},
    publisher     = {arXiv},
    month         = apr,
    doi           = {10.48550/arXiv.1304.6133},
    archiveprefix = {arXiv},
    eprint        = {1304.6133},
    number        = {arXiv:1304.6133},
    primaryclass  = {cs.IT},
    urldate       = {2026-08-24}
}

@inproceedings{anantharam:2013b,
    author    = {Anantharam, Venkat and Gohari, Amin Aminzadeh and
                 Kamath, Sudeep and Nair, Chandra},
    title     = {On Hypercontractivity and the Mutual Information between
                 {{Boolean}} Functions},
    booktitle = {2013 51st {{Annual Allerton Conference}} on {{Communication}},
                 {{Control}}, and {{Computing}} ({{Allerton}})},
    year      = 2013,
    pages     = {13--19},
    doi       = {10.1109/Allerton.2013.6736499},
    month     = oct,
    urldate   = {2026-08-24}
}

@article{barnum:2010,
  title = {Entropy and {{Information Causality}} in {{General Probabilistic Theories}}},
  author = {Barnum, Howard and Barrett, Jonathan and Clark, Lisa Orloff and Leifer, Matthew and Spekkens, Robert and Stepanik, Nicholas and Wilce, Alex and Wilke, Robin},
  year = 2010,
  month = mar,
  journal = {New Journal of Physics},
  volume = {12},
  number = {3},
  eprint = {0909.5075},
  primaryclass = {quant-ph},
  pages = {033024},
  issn = {1367-2630},
  doi = {10.1088/1367-2630/12/3/033024},
  urldate = {2026-10-05},
  archiveprefix = {arXiv}
}

@article{barrett:2007,
    author  = {Barrett, Jonathan},
    title   = {Information Processing in Generalized Probabilistic Theories},
    year    = 2007,
    journal = {Physical Review A},
    volume  = {75},
    number  = {3},
    pages   = {032304},
    month   = mar,
    issn    = {1050-2947, 1094-1622},
    doi     = {10.1103/PhysRevA.75.032304},
    langid  = {english},
    urldate = {2026-01-02}
}

@article{beigi:2015,
    author        = {Beigi, Salman and Gohari, Amin},
    title         = {Monotone {{Measures}} for {{Non-Local Correlations}}},
    year          = 2015,
    journal       = {IEEE Transactions on Information Theory},
    volume        = {61},
    number        = {9},
    pages         = {5185--5208},
    month         = sep,
    issn          = {0018-9448, 1557-9654},
    doi           = {10.1109/TIT.2015.2452253},
    archiveprefix = {arXiv},
    eprint        = {1409.3665},
    primaryclass  = {quant-ph},
    urldate       = {2026-07-24}
}

@article{bell:1964,
    author  = {Bell, J. S.},
    title   = {On the {{Einstein Podolsky Rosen}} Paradox},
    year    = 1964,
    journal = {Physics Physique Fizika},
    volume  = {1},
    number  = {3},
    pages   = {195--200},
    month   = nov,
    issn    = {0554-128X},
    doi     = {10.1103/PhysicsPhysiqueFizika.1.195},
    langid  = {english},
    urldate = {2025-06-24}
}

@article{chaves:2015c,
    author        = {Chaves, Rafael and Majenz, Christian and Gross, David},
    title         = {Information-{{Theoretic Implications}} of
                     {{Quantum Causal Structures}}},
    year          = 2015,
    journal       = {Nature Communications},
    volume        = {6},
    number        = {1},
    pages         = {5766},
    month         = jan,
    issn          = {2041-1723},
    doi           = {10.1038/ncomms6766},
    archiveprefix = {arXiv},
    eprint        = {1407.3800},
    primaryclass  = {quant-ph},
    urldate       = {2024-12-19}
}

@article{collins:2004a,
    author        = {Collins, Daniel and Gisin, Nicolas},
    title         = {A {{Relevant Two Qubit Bell Inequality Inequivalent}} to
                     the {{CHSH Inequality}}},
    year          = 2004,
    journal       = {Journal of Physics A: Mathematical and General},
    volume        = {37},
    number        = {5},
    pages         = {1775--1787},
    month         = feb,
    issn          = {0305-4470, 1361-6447},
    doi           = {10.1088/0305-4470/37/5/021},
    archiveprefix = {arXiv},
    eprint        = {quant-ph/0306129},
    urldate       = {2026-04-09}
}

@article{doherty:2002,
    author        = {Doherty, Andrew C. and Parrilo, Pablo A. and
                     Spedalieri, Federico M.},
    title         = {Distinguishing Separable and Entangled States},
    year          = 2002,
    journal       = {Physical Review Letters},
    volume        = {88},
    number        = {18},
    pages         = {187904},
    month         = apr,
    issn          = {0031-9007, 1079-7114},
    doi           = {10.1103/PhysRevLett.88.187904},
    archiveprefix = {arXiv},
    eprint        = {quant-ph/0112007},
    urldate       = {2026-09-23}
}

@article{doherty:2004,
    author        = {Doherty, Andrew C. and Parrilo, Pablo A. and
                     Spedalieri, Federico M.},
    title         = {A Complete Family of Separability Criteria},
    year          = 2004,
    journal       = {Physical Review A},
    volume        = {69},
    number        = {2},
    pages         = {022308},
    month         = feb,
    issn          = {1050-2947, 1094-1622},
    doi           = {10.1103/PhysRevA.69.022308},
    archiveprefix = {arXiv},
    eprint        = {quant-ph/0308032},
    urldate       = {2026-09-23}
}

@article{fine:1982,
    author    = {Fine, Arthur},
    title     = {Hidden {{Variables}}, {{Joint Probability}}, and the
                 {{Bell Inequalities}}},
    year      = 1982,
    journal   = {Physical Review Letters},
    volume    = {48},
    number    = {5},
    pages     = {291--295},
    month     = feb,
    issn      = {0031-9007},
    doi       = {10.1103/physrevlett.48.291},
    langid    = {english},
    publisher = {American Physical Society (APS)},
    urldate   = {2025-07-22}
}

@misc{gachechiladze:2026,
    title         = {Information {{Causality Characterizes}} the {{Set}} of
                     {{Quantum Correlations}} in the {{Simplest Bell Scenario}}},
    year          = 2026,
    author        = {Gachechiladze, Mariami and Miklin, Nikolai},
    publisher     = {arXiv},
    month         = sep,
    doi           = {10.48550/arXiv.2609.10508},
    archiveprefix = {arXiv},
    eprint        = {2609.10508},
    langid        = {english},
    number        = {arXiv:2609.10508},
    primaryclass  = {quant-ph},
    urldate       = {2026-09-10}
}

@article{gallego:2011,
    author        = {Gallego, Rodrigo and W{\"u}rflinger, Lars Erik and
                     Ac{\'i}n, Antonio and Navascu{\'e}s, Miguel},
    title         = {Quantum Correlations Require Multipartite Information
                     Principles},
    year          = 2011,
    journal       = {Physical Review Letters},
    volume        = {107},
    number        = {21},
    pages         = {210403},
    month         = nov,
    issn          = {0031-9007, 1079-7114},
    doi           = {10.1103/PhysRevLett.107.210403},
    archiveprefix = {arXiv},
    eprint        = {1107.3738},
    primaryclass  = {quant-ph},
    urldate       = {2026-09-23}
}

@article{harder:2013,
    author        = {Harder, Malte and Salge, Christoph and Polani, Daniel},
    title         = {A {{Bivariate Measure}} of {{Redundant Information}}},
    year          = 2013,
    journal       = {Physical Review E},
    volume        = {87},
    number        = {1},
    pages         = {012130},
    month         = jan,
    issn          = {1539-3755, 1550-2376},
    doi           = {10.1103/PhysRevE.87.012130},
    archiveprefix = {arXiv},
    eprint        = {1207.2080},
    primaryclass  = {cs.IT},
    urldate       = {2026-08-28}
}

@article{jain:2024,
    author        = {Jain, Prabhav and Gachechiladze, Mariami and
                     Miklin, Nikolai},
    title         = {Information Causality as a Tool for Bounding the Set of
                     Quantum Correlations},
    year          = 2024,
    journal       = {Physical Review Letters},
    volume        = {133},
    number        = {16},
    pages         = {160201},
    month         = oct,
    issn          = {0031-9007, 1079-7114},
    doi           = {10.1103/PhysRevLett.133.160201},
    archiveprefix = {arXiv},
    eprint        = {2308.02478},
    primaryclass  = {quant-ph},
    urldate       = {2024-11-21}
}

@misc{jain:2026,
    title     = {Bounds on {{Nonlocality}} and {{Random Access Codes}} from
                 {{Extended Information Causality Principle}}},
    year      = 2026,
    author    = {Jain, Prabhav and Miklin, Nikolai and Gachechiladze, Mariami},
    publisher = {arXiv},
    eprint={2606.02416},
    archivePrefix={arXiv},
    doi       = {10.48550/ARXIV.2606.02416},
    urldate   = {2026-06-02}
}

@article{ji:2021,
    author  = {Ji, Zhengfeng and Natarajan, Anand and Vidick, Thomas and
               Wright, John and Yuen, Henry},
    title   = {{{MIP}}* = {{RE}}},
    year    = 2021,
    journal = {Communications of the ACM},
    volume  = {64},
    number  = {11},
    pages   = {131--138},
    month   = nov,
    issn    = {0001-0782, 1557-7317},
    doi     = {10.1145/3485628},
    langid  = {english},
    urldate = {2026-09-23}
}

@article{lang:2014,
  title = {Closed Sets of Correlations: Answers from the Zoo},
  shorttitle = {Closed Sets of Correlations},
  author = {Lang, Ben and Vertesi, Tamas and Navascues, Miguel},
  year = 2014,
  month = oct,
  journal = {Journal of Physics A: Mathematical and Theoretical},
  volume = {47},
  number = {42},
  eprint = {1402.2850},
  primaryclass = {quant-ph},
  pages = {424029},
  issn = {1751-8113, 1751-8121},
  doi = {10.1088/1751-8113/47/42/424029},
  urldate = {2026-10-05},
  archiveprefix = {arXiv}
}

@article{lesniewski:1999,
  title = {Monotone {{Riemannian Metrics}} and {{Relative Entropy}} on {{Non-Commutative Probability Spaces}}},
  author = {Lesniewski, Andrew and Ruskai, Mary Beth},
  year = 1999,
  month = nov,
  journal = {Journal of Mathematical Physics},
  volume = {40},
  number = {11},
  eprint = {math-ph/9808016},
  pages = {5702--5724},
  issn = {0022-2488, 1089-7658},
  doi = {10.1063/1.533053},
  urldate = {2026-10-05},
  archiveprefix = {arXiv}
}

@article{masanes:2006a,
  title = {General Properties of {{Nonsignaling Theories}}},
  author = {Masanes, Ll and Acin, A. and Gisin, N.},
  year = 2006,
  month = jan,
  journal = {Physical Review A},
  volume = {73},
  number = {1},
  eprint = {quant-ph/0508016},
  pages = {012112},
  issn = {1050-2947, 1094-1622},
  doi = {10.1103/PhysRevA.73.012112},
  urldate = {2026-10-05},
  archiveprefix = {arXiv}
}

@article{miklin:2021,
    author        = {Miklin, Nikolai and Paw{\l}owski, Marcin},
    title         = {Information {{Causality}} without Concatenation},
    year          = 2021,
    journal       = {Physical Review Letters},
    volume        = {126},
    number        = {22},
    pages         = {220403},
    month         = jun,
    issn          = {0031-9007, 1079-7114},
    doi           = {10.1103/PhysRevLett.126.220403},
    archiveprefix = {arXiv},
    eprint        = {2101.12710},
    primaryclass  = {quant-ph},
    urldate       = {2024-11-21}
}

@misc{miklin:2026,
    title         = {Communication Complexity Bounds from Information Causality},
    year          = 2026,
    author        = {Miklin, Nikolai and Jain, Prabhav and
                     Gachechiladze, Mariami},
    publisher     = {arXiv},
    month         = feb,
    doi           = {10.48550/arXiv.2602.10206},
    archiveprefix = {arXiv},
    eprint        = {2602.10206},
    number        = {arXiv:2602.10206},
    primaryclass  = {quant-ph},
    urldate       = {2026-02-12}
}

@article{navascues:2008,
    author        = {Navascues, Miguel and Pironio, Stefano and Acin, Antonio},
    title         = {A Convergent Hierarchy of Semidefinite Programs
                     Characterizing the Set of Quantum Correlations},
    year          = 2008,
    journal       = {New Journal of Physics},
    volume        = {10},
    number        = {7},
    pages         = {073013},
    month         = jul,
    issn          = {1367-2630},
    doi           = {10.1088/1367-2630/10/7/073013},
    archiveprefix = {arXiv},
    eprint        = {0803.4290},
    primaryclass  = {quant-ph},
    urldate       = {2024-12-27}
}

@article{pawlowski:2009,
    author        = {Pawlowski, M. and Paterek, T. and Kaszlikowski, D. and
                     Scarani, V. and Winter, A. and Zukowski, M.},
    title         = {Information {{Causality}} as a {{Physical Principle}}},
    year          = 2009,
    journal       = {Nature},
    volume        = {461},
    number        = {7267},
    pages         = {1101--1104},
    month         = oct,
    issn          = {0028-0836, 1476-4687},
    doi           = {10.1038/nature08400},
    archiveprefix = {arXiv},
    eprint        = {0905.2292},
    primaryclass  = {quant-ph},
    urldate       = {2024-11-21}
}

@misc{pollyceno:2026,
    title         = {Communication-Constrained Nonlocal Correlations},
    year          = 2026,
    author        = {Pollyceno, Lucas and Freudenheim, Denis and
                     Nogueira, Jos{\'e} and Chaturvedi, Anubhav and
                     Rabelo, Rafael and Paw{\l}owski, Marcin},
    publisher     = {arXiv},
    month         = mar,
    doi           = {10.48550/arXiv.2603.08610},
    archiveprefix = {arXiv},
    eprint        = {2603.08610},
    number        = {arXiv:2603.08610},
    primaryclass  = {quant-ph},
    urldate       = {2026-03-11}
}

@book{polyanskiy:2025,
    author     = {Polyanskiy, Yury and Wu, Yihong},
    title      = {Information {{Theory}}: {{From Coding}} to {{Learning}}},
    year       = 2025,
    publisher  = {Cambridge University Press},
    edition    = {1},
    isbn       = {978-1-108-96635-1 978-1-108-83290-8},
    doi        = {10.1017/9781108966351},
    month      = jan,
    shorttitle = {Information {{Theory}}},
    urldate    = {2026-08-24}
}

@article{popescu:1994,
    author  = {Popescu, Sandu and Rohrlich, Daniel},
    title   = {Quantum Nonlocality as an Axiom},
    year    = 1994,
    journal = {Foundations of Physics},
    volume  = {24},
    number  = {3},
    pages   = {379--385},
    month   = mar,
    issn    = {0015-9018, 1572-9516},
    doi     = {10.1007/BF02058098},
    langid  = {english},
    urldate = {2024-12-24}
}

@article{raginsky:2016,
    author  = {Raginsky, Maxim},
    title   = {Strong {{Data Processing Inequalities}} and {$\Phi$}--{{Sobolev}}
               Inequalities for {{Discrete Channels}}},
    year    = 2016,
    journal = {IEEE Transactions on Information Theory},
    volume  = {62},
    number  = {6},
    pages   = {3355--3389},
    month   = jun,
    issn    = {0018-9448, 1557-9654},
    doi     = {10.1109/TIT.2016.2549542},
    urldate = {2026-09-30}
}

@article{hirche:2022,
  title = {On Contraction Coefficients, Partial Orders and Approximation of Capacities for Quantum Channels},
  author = {Hirche, Christoph and Rouz{\'e}, Cambyse and Fran{\c c}a, Daniel Stilck},
  year = 2022,
  month = nov,
  journal = {Quantum},
  volume = {6},
  eprint = {2011.05949},
  primaryclass = {quant-ph},
  pages = {862},
  issn = {2521-327X},
  doi = {10.22331/q-2022-11-28-862},
  urldate = {2026-10-05},
  archiveprefix = {arXiv}
}

@article{plavala:2023,
  title = {General Probabilistic Theories: {{An}} Introduction},
  shorttitle = {General Probabilistic Theories},
  author = {Pl{\'a}vala, Martin},
  year = 2023,
  month = sep,
  journal = {Physics Reports},
  volume = {1033},
  eprint = {2103.07469},
  primaryclass = {quant-ph},
  pages = {1--64},
  issn = {03701573},
  doi = {10.1016/j.physrep.2023.09.001},
  urldate = {2024-12-24},
  archiveprefix = {arXiv}
}

@article{short:2010a,
  title = {Entropy in General Physical Theories},
  author = {Short, Anthony J. and Wehner, Stephanie},
  year = 2010,
  month = mar,
  journal = {New Journal of Physics},
  volume = {12},
  number = {3},
  eprint = {0909.4801},
  primaryclass = {quant-ph},
  pages = {033023},
  issn = {1367-2630},
  doi = {10.1088/1367-2630/12/3/033023},
  urldate = {2026-10-05},
  archiveprefix = {arXiv}
}

@article{slofstra:2019,
    author  = {Slofstra, William},
    title   = {The set of quantum correlations is not closed},
    year    = 2019,
    journal = {Forum of Mathematics, Pi},
    volume  = {7},
    pages   = {e1},
    issn    = {2050-5086},
    doi     = {10.1017/fmp.2018.3},
    langid  = {english},
    urldate = {2026-09-23}
}

@article{weilenmann:2020,
    author        = {Weilenmann, Mirjam and Colbeck, Roger},
    title         = {Analysing Causal Structures in Generalised Probabilistic
                     Theories},
    year          = 2020,
    journal       = {Quantum},
    volume        = {4},
    pages         = {236},
    month         = feb,
    issn          = {2521-327X},
    doi           = {10.22331/q-2020-02-27-236},
    archiveprefix = {arXiv},
    eprint        = {1812.04327},
    primaryclass  = {quant-ph},
    urldate       = {2025-11-27}
}

@article{yu:2026,
  title = {Information Causality beyond the Random Access Code Model},
  author = {Yu, Baichu and Scarani, Valerio},
  year = 2026,
  month = jul,
  journal = {Physical Review A},
  volume = {114},
  number = {1},
  eprint = {2201.08986},
  primaryclass = {quant-ph},
  pages = {012202},
  issn = {2469-9926, 2469-9934},
  doi = {10.1103/s52j-2jr7},
  urldate = {2026-10-05},
  archiveprefix = {arXiv}
}

@article{brassard:2006,
  title = {Limit on Nonlocality in Any World in Which Communication Complexity Is Not Trivial},
  author = {Brassard, Gilles and Buhrman, Harry and Linden, Noah and Methot, Andre A. and Tapp, Alain and Unger, Falk},
  year = 2006,
  month = jun,
  journal = {Physical Review Letters},
  volume = {96},
  number = {25},
  eprint = {quant-ph/0508042},
  pages = {250401},
  issn = {0031-9007, 1079-7114},
  doi = {10.1103/PhysRevLett.96.250401},
  urldate = {2026-10-05},
  archiveprefix = {arXiv}
}

@article{linden:2007,
  title = {Quantum Nonlocality and Beyond: Limits from Nonlocal Computation},
  author = {Linden, Noah and Popescu, Sandu and Short, Anthony J. and Winter, Andreas},
  journal = {Phys. Rev. Lett.},
  volume = {99},
  issue = {18},
  pages = {180502},
  numpages = {4},
  year = {2007},
  month = {Oct},
  publisher = {American Physical Society},
  doi = {10.1103/PhysRevLett.99.180502},
  url = {https://link.aps.org/doi/10.1103/PhysRevLett.99.180502}
}

@article{navascues:2010,
  title = {A Glance beyond the Quantum Model},
  author = {Navascues, Miguel and Wunderlich, Harald},
  year = 2010,
  month = mar,
  journal = {Proceedings of the Royal Society A: Mathematical, Physical and Engineering Sciences},
  volume = {466},
  number = {2115},
  eprint = {0907.0372},
  primaryclass = {quant-ph},
  pages = {881--890},
  issn = {1364-5021, 1471-2946},
  doi = {10.1098/rspa.2009.0453},
  urldate = {2025-06-10},
  archiveprefix = {arXiv}
}

@article{fritz:2013,
  title = {Local Orthogonality as a Multipartite Principle for Quantum Correlations},
  author = {Fritz, T. and Sainz, A. B. and Augusiak, R. and Brask, J. B. and Chaves, R. and Leverrier, A. and Ac{\'i}n, A.},
  year = 2013,
  month = aug,
  journal = {Nature Communications},
  volume = {4},
  number = {1},
  pages = {2263},
  eprint = {1210.3018},
  primaryclass = {quant-ph},
  issn = {2041-1723},
  doi = {10.1038/ncomms3263},
  archiveprefix = {arXiv},
  url = {https://doi.org/10.1038/ncomms3263},
}

@Article{cirelson:1993,
  author = 	 {B. S. Tsirelson},
  title = 	 {Some results and problems on quantum {B}ell-type inequalities},
  journal = 	 {Hadronic Journal Supplement},
  year = 	 {1993},
  OPTkey = 	 {},
  volume = 	 {8},
  OPTnumber = 	 {},
  pages = 	 {329--345},
  OPTmonth = 	 {},
  OPTnote = 	 {},
  OPTannote = 	 {}
}

@article{navascues:2015,
  title = {Almost Quantum Correlations},
  author = {Navascu{\'e}s, Miguel and Guryanova, Yelena and Hoban, Matty J. and Ac{\'i}n, Antonio},
  year = 2015,
  month = feb,
  journal = {Nature Communications},
  volume = {6},
  number = {1},
  eprint = {1403.4621},
  primaryclass = {quant-ph},
  pages = {6288},
  issn = {2041-1723},
  doi = {10.1038/ncomms7288},
  urldate = {2026-10-05},
  archiveprefix = {arXiv}
}

@misc{marco:costantino,
    title         = {Bounding single-copy information causality violations},
    year          = 2026,
    author        = {Pompili, Marco and Budroni, Costantino},
    publisher     = {In Preparation},
    month         = sep
}

@misc{hardy:2001,
      title={Quantum Theory From Five Reasonable Axioms}, 
      author={Lucien Hardy},
      year={2001},
      eprint={quant-ph/0101012},
      archivePrefix={arXiv},
      primaryClass={quant-ph},
      url={https://arxiv.org/abs/quant-ph/0101012}, 
}

@article{chiribella:2010,
  title = {Probabilistic Theories with Purification},
  author = {Chiribella, G. and D'Ariano, G. M. and Perinotti, P.},
  year = 2010,
  month = jun,
  journal = {Physical Review A},
  volume = {81},
  number = {6},
  eprint = {0908.1583},
  primaryclass = {quant-ph},
  pages = {062348},
  issn = {1050-2947, 1094-1622},
  doi = {10.1103/PhysRevA.81.062348},
  urldate = {2026-10-05},
  archiveprefix = {arXiv}
}

@article{janotta:2013,
  title = {Generalized {{Probabilistic Theories Without}} the {{No-Restriction Hypothesis}}},
  author = {Janotta, Peter and Lal, Raymond},
  year = 2013,
  month = may,
  journal = {Physical Review A},
  volume = {87},
  number = {5},
  eprint = {1302.2632},
  primaryclass = {quant-ph},
  pages = {052131},
  issn = {1050-2947, 1094-1622},
  doi = {10.1103/PhysRevA.87.052131},
  urldate = {2026-10-05},
  archiveprefix = {arXiv}
}

\appendix
\numberwithin{proposition}{section}
\renewcommand{\theproposition}{\thesection\arabic{proposition}}
\numberwithin{lemma}{section}
\renewcommand{\thelemma}{\thesection\arabic{lemma}}

\begin{widetext}

\section{The ghostworld generalized probabilistic theories}
    \label{app:ghostworld_GPT}

    In this appendix, we derive properties of the ghostworld sets of correlations $\GW_n$. We further introduce in more detail the bipartite $n$-ghostworld GPTs and establish some
    of their structural properties that are used in our proof of
    Result~\ref{res:gw-ic}. 

    \subsection*{Ghostworld sets of correlations \texorpdfstring{$\GW_n$}{GWn}}

    We start by noticing that the positivity, non-signaling, symmetry,
    normalization, and marginal conditions on $Q_A,Q_B$, are finitely many
    linear constraints. Their feasible set is a bounded polytope. Projecting
    onto their bipartite marginal $P$ shows that $\GW_n^{(m,d)}$ is a compact convex
    polytope for all $m,d$. We further have that
    \beq
        \label{eq:gw-hierarchy}
        \LHV \subseteq \GW_{n+1}
             \subseteq \GW_n
             \subseteq \cdots
             \subseteq \GW_0
             = \NS,
    \eeq
    where $\NS$ denotes the set of all non-signalling correlations. 
    Moreover, for any Bell scenario $\mathsf B = (m_A, m_B, d_A, d_B)$, this
    hierarchy converges to the local set at a finite level. For completeness, we
    provide a proof of this fact.
    \begin{proposition}[Convergence at finite level]
        \label{prop:gw-finite-locality}
        In a Bell scenario with $m_A,m_B$ inputs, 
        \beq
            \label{eq:gw-finite-locality}
            \GW_n = \LHV \quad\text{if}\quad n + 1 \geq \min\{m_A,m_B\}
        \eeq
        More generally, every restriction of a behavior in $\GW_n$ to at most
        $n + 1$ inputs for either party is local.
    \end{proposition}

    The proof is based on Fine's
        theorem~\cite{fine:1982} and is analogous to the shareability result of~\cite[Result~1]{masanes:2006a}.
    
    \begin{proof}
    W.l.o.g.\ let $m_A \leq m_B$. Now let us first consider the case $m_A = n + 1$. Let there be a ghostworld behaviour $P=\{p(a,b|x,y)\}$ and label Alice's measurements
        $x \in \{0,\ldots,n\}$. Then, in the extension $Q_A$, set $x_j = j$ for the
        $j$-th copy, and let $\lambda = (a_0,\ldots,a_n)$. The distribution
        $p(\lambda) = \sum_b q_A(\lambda,b|0,\ldots,n,y)$ is independent of $y$
        by non-signalling.  For all $\lambda$ such that
        $p(\lambda) > 0$, define 
        \beq
            p(a|j,\lambda) = \delta_{a,a_j},\qquad p(b|y,\lambda)
                           = \frac{q_A(\lambda,b|0,\ldots,n,y)}{p(\lambda)} \qquad \forall a,j,b,y.
        \eeq
        For $\lambda$ such that  $p(\lambda) = 0$, set
        $p(a|j,\lambda) = \delta_{a,a_j} \forall a,j$  and choose any normalized probability
        distribution $p(\cdot|y,\lambda)$ for each $y$, since the corresponding
        term in the local hidden-variable decomposition below vanishes. Then
        \beq
            \sum_\lambda p(\lambda)p(a|j,\lambda)p(b|y,\lambda) = p(a,b|j,y), \ \forall a,b,j,y
        \eeq
        i.e., there is a local model reproducing $P$.
        The same construction applies whenever $m_A < n + 1$, after
        marginalizing over the additional systems in $Q_A$. 
    \end{proof}

    We further recall from~\cite{lang:2014}, that any ghostworld set of correlations is closed under wirings and shared randomness.

    \begin{proposition}[Closure under wirings~\cite{lang:2014}]
        \label{prop:gw-wirings}
        For any $m,d$ and any $n$, let $P_1,\ldots,P_c \in \GW_n^{(m,d)}$ be independent resources. Every behavior
        obtained from them by local wirings and shared randomness belongs to
        $\GW_n$ in the resulting Bell scenario.
    \end{proposition}

    \subsection*{The generalised probabilistic theories of \texorpdfstring{$n$}{n}-ghostworld and their measurement entropy}
   Let us specify our
    bipartite $n$-ghostworld GPTs, introduced in Section~\ref{sec:preliminary}, in more detail. An elementary system $A$ has the boxworld state
    space of all normalized conditional distributions
    $\omega_A = \{p(a|x)\}_{a,x}$. In the ordered vector space
    formalism~\cite{plavala:2023}, one usually considers a non-redundant coordinate
    system, e.g.~$\omega_A = (p(0|0),p(0|1),1)$ in the two-input and two-output
    case, but it suffices to keep the redundant parametrization for our
    purposes. The effect space
    for $A$ is also equal to the boxworld one. The ray-extremal effects on $A$ are denoted $\{\delta_{a|x}\}$ and
    they pick out one component of the state, $\delta_{a|x}(\omega_A) = p(a|x)$; 
    their sum over $a$ is the unit effect $u_A$ for all $x$.  For a pair $AB$ of elementary
    systems, we restrict the joint states to $\Omega_{AB}^{(n)} = \GW_n$. The
    allowed joint effects and measurements are those of boxworld, i.e., $\mathcal E_{AB}^{(n)} = \big\{ e\big|_{\Omega_{AB}^{(n)}} \colon e \in \mathcal E_{AB}^{\mathrm{BW}} \big\}$. 
     Note that we are restricting the states
    \emph{without} enlarging the allowed effects to all affine functionals from
    $\Omega^{(n)}_{AB}$ to $[0,1]$, hence our bipartite ghostworld GPTs do not satisfy the
    no-restriction hypothesis~\cite{chiribella:2010, janotta:2013}.

    In general, when dealing with GPTs in the context of information processing, it is convenient to describe classical information and GPT states in a joint system (analogously to classical-quantum states).  One can consider classical-GPT states the usual way via direct sums of the    involved probability distributions. Equivalently, a classical-GPT state is a    family of subnormalized GPT states whose total normalization is $1$. A joint    measurement on the classical-GPT state allows to read the classical register
    and use the outcome to choose a measurement on the GPT system.   
    Notice that for our purposes it is not necessary to generalise $n$-ghostworld to more parties. 

    Every effect on an elementary system can be represented by a nonnegative
    linear combination of ray extremal 
    effects~\cite[Appendix~D]{barrett:2007},
    namely $e = \sum_{a,x} c_{ax} \delta_{a|x}$ with $c_{ax} \ge 0$. Multipartite boxworld states are normalized non-signaling behaviors, and
    their effects can be represented as nonnegative combinations of products of 
    ray extremal effects~\cite{barrett:2007},
    \beq
        e = \sum_{\mathbf a,\mathbf x} c_{\mathbf a\mathbf x} \delta_{a_1|x_1}
            \otimes \cdots
            \otimes \delta_{a_k|x_k}, \qquad c_{\mathbf a\mathbf x}
          \ge 0. \label{eq:gw-multipartite-effects}
    \eeq
    Here $\mathbf a = (a_1,\ldots,a_k)$ and $\mathbf x = (x_1,\ldots,x_k)$.

    We now characterize the fine-grained measurements and the corresponding
    measurement entropies that we will need in the IC proofs.

    \begin{lemma}[Elementary fine-grained measurements]
        \label{lem:gw-fiducial-measurements}
        Up to splitting proportional outcomes, every fine-grained measurement
        $M \in \mathcal M^*$ on an elementary boxworld system has the form
        \beq
            M = \{c_x\delta_{a|x}\}_{x,a},\qquad c_x
              \geq 0, \qquad\sum_x c_x
              = 1. \label{eq:gw-fiducial-measurements}
        \eeq
    \end{lemma}
    \begin{proof}
        A fine-grained measurement is by definition composed
        of indecomposable effects, which are those that lie on an extremal ray of the effect cone. 
        Thus, after merging proportional outcomes, we
        can write any fine-grained measurement as $\{c_{ax}\delta_{a|x}\}_{a,x}$
        with $c_{ax} \ge 0$, with null coefficients for absent effects.
        Normalization requires, for every state $\omega_A = \{p(a|x)\}_{a,x}$,
        \beq
            \sum_{a,x} c_{ax} p(a|x) = 1. \label{eq:gw-fiducial-normalization}
        \eeq
        Applying it to all the deterministic states $\tilde{\omega}_A=\{p(a|x)\}$ such that $p(a|x) = \delta_{a,f(x)} \ \forall a,x$,
        this gives
        \beq
            \sum_x c_{f(x),x} = 1 \qquad \text{for every deterministic } f(x). \label{eq:gw-fiducial-deterministic}
        \eeq
        There are always deterministic states differing from $\tilde{\omega}_A$ only in the $\bar x$ measurement, thus varying $f(\bar x)$ while keeping the other outcomes fixed shows that
        $c_{a\bar x}$ is independent of $a$. Since $\bar x$ is arbitrary,
        $c_{ax} = c_x$, and normalization gives $\sum_x c_x = 1$, proving
        Eq.~\eqref{eq:gw-fiducial-measurements}.
    \end{proof}

    \begin{proposition}[Measurement entropies]
        \label{prop:gw-measurement-entropies}
        Let $A,B$ be elementary boxworld systems and $C,M$ classical registers.
        For states on $A$, $AB$, and $CMB$, respectively,
        \begin{subequations}
            \label{eq:gw-local-conditional-formulas}
            \begin{align}
                H_+(A)          & = \min_x H(A_x), \label{eq:gw-elementary-entropy}                                    \\
                H_+(B \mid A)   & = \min_x\sum_a p(a|x)\min_y H(B_y \mid A_x = a), \label{eq:gw-conditional-entropy}   \\
                H_+(C \mid M,B) & = \sum_m p(m)\min_y H(C \mid M = m,B_y), \label{eq:gw-classical-conditional-entropy}
            \end{align}
        \end{subequations}
        where $A_x$ and $B_y$ denote the outcomes of local measurements labelled by $x$ and $y$ in the probability table representation we use for the states,
        e.g., $H(A_x) \coloneqq H(p(\cdot|x))$ with $H$ the Shannon
        entropy. In Eq.~\eqref{eq:gw-conditional-entropy}, the optimal $y$ may
        depend on $x$ and $a$; in
        Eq.~\eqref{eq:gw-classical-conditional-entropy}, the optimal $y$ may
        depend on $m$.
    \end{proposition}

    \begin{proof}
        Notice that proportional splitting of outcomes cannot lower
        their Shannon entropy. By Lemma~\ref{lem:gw-fiducial-measurements}, we
        can therefore write
        \beq
            \begin{aligned}
                H_+(A) & = \min_{\{c_x\}} H(\{c_xp(a|x)\}_{a,x})                        \\
                       & = \min_{\{c_x\}}\left[ H(\{c_x\}) + \sum_x c_x H(A_x) \right],
            \end{aligned} \label{eq:gw-entropy-simplex}
        \eeq
        where $c_x \geq 0$ and $\sum_xc_x = 1$. The second equality follows by
        writing the definition of Shannon entropy and regrouping terms. This is
        a concave function on the probability simplex, so its minimum is
        attained at a vertex $c_x = \delta_{x,\bar x}$. Minimizing over $\bar x$
        gives Eq.~\eqref{eq:gw-elementary-entropy}.

        For conditional measurement entropy, the measurement on the conditioning
        system may be taken fine-grained~\cite[Appendix~A]{weilenmann:2020}. This mean that on
        $A$, we can take a measurement of the form $\{c_x\delta_{a|x}\}_{a,x}$, up to
        proportional splitting. The outcome $(x,a)$ has probability $c_xp(a|x)$
        and (if nonzero) it leaves $B$ in the following state, independent of
        $c_x$:
        \beq
            \omega_{B}^{a|x} = \left\{ \frac{p(a,b|x,y)}{p(a|x)} \right\}_{b,y}. \label{eq:gw-fiducial-conditional-state}
        \eeq
        Hence
        \beq
            H_+(B \mid A) = \min_{\{c_x\}}\sum_{x,a} c_xp(a|x)H_+(\omega_{B}^{a|x}), \label{eq:gw-conditional-simplex}
        \eeq
        with $c_x \geq 0$ and $\sum_xc_x = 1$. The objective is linear in these
        coefficients, so its minimum is attained at a vertex, i.e., at a fixed
        input $x$. Applying Eq.~\eqref{eq:gw-elementary-entropy} separately to
        each conditional state of $B$ gives
        Eq.~\eqref{eq:gw-conditional-entropy}.

        Finally, a measurement on $M B$ may be described by reading the value of
        the classical $M$ first, followed by a measurement on $B$ chosen for
        that value~\cite[Appendix~A]{weilenmann:2020}. Up to proportional
        splitting, the fine-grained measurements on $MB$ therefore have effects
        $\{c_{my}\delta_m \otimes \delta_{b|y}\}_{m,b,y}$, where $\delta_m$
        reads the classical value $m$ and $c_{my} \geq 0$, $\sum_yc_{my} = 1$
        for each $m$. Since $C$ is classical, the resulting 
        entropy is
        \beq
         H_+(C \mid M, B)= \min_{\{c_{my}\}}  \sum_m p(m)\sum_y c_{my}H(C \mid M = m,B_y). \label{eq:gw-classical-conditional-simplex}
        \eeq
        This expression is linear in each vector $(c_{my})_y$, which can be chosen
        independently for every $m$. Hence the minimum is attained at a vertex of each simplex, which 
        gives Eq.~\eqref{eq:gw-classical-conditional-entropy}. %
    \end{proof}

    \subsection*{
        The \texorpdfstring{$I_{mm22}$}{Imm22} ghostworld bound
    }

    In this section, we prove the optimality of the points $P_m$ in 
    Section~\ref{sec:insufficiency-finite-length} for the $I_{mm22}$ 
    Bell functional.
    We use the notation $p_x^A = p_A(0|x)$,
    $p_y^B = p_B(0|y)$, $p_{xy} = p(0,0|x,y)$. In the Collins--Gisin
    notation~\cite{collins:2004a}, transposed so that Alice's inputs index rows
    and Bob's inputs index columns, with the marginal coefficients/probabilities in the first column and row, we have
    \beq
        I_{mm22}
        = \left( \begin{behaviorarray}{c|ccccc}
                            & -(m-1) & -(m-2) & \cdots                & -1     & 0      \\
                     \hline
                     -1     & 1      & 1      & \cdots                & 1      & 1      \\
                     0      & 1      & 1      & \cdots                & 1      & -1     \\
                     0      & 1      & 1      & \cdots                & -1     & 0      \\
                     \vdots & \vdots & \vdots & \reflectbox{$\ddots$} & \vdots & \vdots \\
                     0      & 1      & -1     & \cdots                & 0      & 0
                 \end{behaviorarray} \right), \qquad P_m
        = \frac1m\left( \begin{behaviorarray}{c|ccccc}
                                   & 1      & 1      & \cdots                & 1      & 1      \\
                            \hline
                            m-1    & 1      & 1      & \cdots                & 1      & 1      \\
                            m-1    & 1      & 1      & \cdots                & 1      & 0      \\
                            m-1    & 1      & 1      & \cdots                & 0      & 1      \\
                            \vdots & \vdots & \vdots & \reflectbox{$\ddots$} & \vdots & \vdots \\
                            m-1    & 1      & 0      & \cdots                & 1      & 1
                        \end{behaviorarray} \right). \label{eq:gw-collins-gisin-tables}
    \eeq

    The local bound of the $I_{mm22}$ functional is zero for all
    $m \ge 2$~\cite{collins:2004a}. For a behavior in $\GW_{m-2}$ we notice
    that, by Proposition~\ref{prop:gw-finite-locality}, deleting any input on
    Alice or Bob leaves a local sub-behavior. This yields the following two
    bounds.

    \begin{lemma}[Bob deletions]
        \label{lem:bell-deletions}
        For $P \in \GW_{m-2}$ in the $(m,m,2,2)$ scenario,
        $(m - 1)I_{mm22}(P) \leq p_1^A$.
    \end{lemma}

    \begin{proof}
        Define the contribution of Bob's $k$th input, and the functional that
        ignores that contribution, as
        \beq
            C_k \coloneqq \sum_xM_{xk}p_{xk}
                          - (m - k)p_k^B, \qquad I^{B \setminus k}(P)
                \coloneqq I_{mm22}(P) - C_k.
        \eeq
        The functional $I^{B \setminus k}$ has local bound zero: any local model
        on the remaining inputs can be completed with the deterministic outcome
        $b = 1$ at input $k$, so that $p_k^B = p_{xk} = 0$ and the full
        functional $I_{mm22}$ has exactly the same value. Moreover,
        $I^{B \setminus k}(P) \le 0$, as the sub-behavior obtained from $P$ by
        deleting Bob input $k$ is local. Summing this inequality over all $k$,
        and using $\sum_k C_k = I_{mm22}(P) + p_1^A$ that comes from the
        definition of $I_{mm22}$, gives
        \beq
            0 \geq \sum_{k=1}^{m}I^{B \setminus k}(P)
              = (m - 1)I_{mm22}(P) - p_1^A,
        \eeq
        as we wanted to prove.
    \end{proof}

    \begin{lemma}[Alice completions]
        \label{lem:bell-completion}
        For $P \in \GW_{m-2}$ in the $(m,m,2,2)$ scenario,
        $I_{mm22}(P) \leq 1 - p_1^A$.
    \end{lemma}

    \begin{proof}
        Delete Alice's first measurement setting, getting a local sub-behavior
        in $(m - 1,m,2,2)$. Complete it by assigning $a = 0$ to the deleted
        input. The resulting local behavior $Q$ satisfies $q_1^{A} = 1$ and
        $q_{1y} = p_y^B$, with all other entries unchanged. Since the first row
        of $M_{xy}$ of the $I_{mm22}$ functional consists of ones,
        \beq
            \begin{aligned}
                0 \geq I_{mm22}(Q) & = I_{mm22}(P) - \sum_y p_{1y} + \sum_y p_y^B + p_1^A - 1 \\
                                   & = I_{mm22}(P) + p_1^A - 1 + \sum_y p(1,0|1,y).
            \end{aligned}
        \eeq
        Therefore,
        $I_{mm22}(P) \le 1 - p_1^A - \sum_y p(1,0|1,y) \le 1 - p_1^A$.
    \end{proof}

    Adding the bounds in Lemmas~\ref{lem:bell-deletions}
    and~\ref{lem:bell-completion} gives
    \beq
        m I_{mm22}(P) \leq p_1^A
                           + (1 - p_1^A), \quad \text{i.e.,} \quad I_{mm22}(P)
                      \le \frac{1}{m}. \label{eq:gw-bell-upper-combination}
    \eeq
    This proves the upper bound in Proposition~\ref{prop:bell-upper},
    and thus the optimality of $P_m$ for the $I_{mm22}$ functional.

    \section{Contraction coefficients and ghostworld points}
    \label{app:contraction}

    In this appendix, we provide more details on the proof of Result~\ref{res:contraction}.

\subsection*{Proof of Theorem~\ref{thm:contraction-criterion}}
        Let $P$ satisfy the contraction criterion, and take an arbitrary
        deterministic single-copy protocol with choices
        $x = f_X(\rvec)$, $y = f_Y(s,m')$, decoding $g = f_G(s,m',b)$, and
        communication channel $p(m'|\rvec,a)$. Write $x_\rvec = f_X(\rvec)$. If
        $p(\rvec,m') > 0$, the distributions of Alice's and Bob's outcomes
        before and after conditioning on $m'$ have components
        \beq
            \begin{aligned}
                p_{A|\rvec}(a)   & = p_A(a|x_\rvec), & p_{A|\rvec,m'}(a)   & = \frac{p_A(a|x_\rvec)p(m'|\rvec,a)}{p(m'|\rvec)}, \\
                p_{B_y|\rvec}(b) & = p_B(b|y),       & p_{B_y|\rvec,m'}(b) & = \sum_a p_{A|\rvec,m'}(a)W_{x_\rvec y}(b|a).
            \end{aligned} \label{eq:contraction-posteriors}
        \eeq
        If $p(\rvec,m') = 0$, set $p_{A|\rvec,m'} \coloneqq p_{A|\rvec}$ and $p_{B_y|\rvec,m'} \coloneqq p_{B_y|\rvec}$ instead (so that the following equation
        reads $0 \le 0$). 
        Here, the expression for $p_{A|\rvec,m'}$ follows from Bayes' rule, and
        $p_{B_y|\rvec}(b) = p_B(b|y)$ follows from non-signaling. Applying
        Eq.~\eqref{eq:relative-entropy-criterion} to these distributions gives
        \beq
            \DKL\left( \bm p_{B_y|\rvec,m'} \,\middle\|\, \bm p_{B_y|\rvec} \right) \leq \eta_{x_\rvec y} \DKL\left( \bm p_{A|\rvec,m'} \,\middle\|\, \bm p_{A|\rvec} \right). \label{eq:contraction-protocol-kl}
        \eeq
        Averaging over $m'$ at fixed $\rvec$, we recognize the conditional
        mutual informations
        \beq
            \begin{aligned}
                \sum_{m'}p(m'|\rvec) \DKL\left( \bm p_{B_y|\rvec,m'} \,\middle\|\, \bm p_{B_y|\rvec} \right) & = I(M':B_y \mid \Rvec = \rvec), \\
                \sum_{m'}p(m'|\rvec) \DKL\left( \bm p_{A|\rvec,m'} \,\middle\|\, \bm p_{A|\rvec} \right)     & = I(M':A \mid \Rvec = \rvec).
            \end{aligned} \label{eq:contraction-kl-averages}
        \eeq
        Summing over $y$ and using $\sum_{y} \eta_{x_\rvec y} \le 1$, we get
        \beq
            \begin{aligned}
                \sum_y I(M':B_y \mid \Rvec = \rvec) %
                \leq I(M':A \mid \Rvec = \rvec).
            \end{aligned} \label{eq:contraction-protocol-bound}
        \eeq
        Averaging over $\rvec$ therefore gives
        $\sum_y I(M':B_y \mid \Rvec) \leq I(M':A \mid \Rvec)$. Moreover,
        non-signaling implies $I(\Rvec:B_y) = 0$, so
        \beq
            \begin{aligned}
                I(M':B_y \mid \Rvec) & = I(\Rvec,M':B_y)                  \\
                                     & = I(M':B_y) + I(\Rvec:B_y \mid M') \\
                                     & \geq I(\Rvec:B_y \mid M').
            \end{aligned} \label{eq:contraction-nonsignaling-chain}
        \eeq
        Combining these bounds, adding $I(\Rvec:M')$, and using the chain rule
        yields
        \beq
            \begin{aligned}
                I(\Rvec:M') + \sum_y I(\Rvec:B_y \mid M') & \leq I(\Rvec:M') + I(M':A \mid \Rvec) \\
                                                          & = I(\Rvec,A:M').
            \end{aligned} \label{eq:contraction-conditional-information}
        \eeq
        This proves Eq.~\eqref{eq:middle-ic-inequality-operational}, the
        intermediate inequality in the statement of the theorem.

        To obtain IC, write $B_s^\pi$ for Bob's outcome when Bob's question
        is $s$, with measurement $y = f_Y(s,m')$. 
        Since $g_s = f_G(s,m',b_s^\pi)$, data processing gives
        \beq
            \sum_{s=0}^{n-1}I(R_s:G_s \mid \Rvec_{<s}) \leq \sum_{s=0}^{n-1}I(R_s:M',B_s^\pi \mid \Rvec_{<s}). \label{eq:contraction-decoding-data-processing}
        \eeq
        Fix a message value $m'$ of positive probability and partition Bob's $s$ indices according to his measurement:
        \beq
            \mathcal I_y(m') \coloneqq \{s \in [n] : f_Y(s,m') = y\}. \label{eq:contraction-decoding-groups}
        \eeq
        For $s \in \mathcal I_y(m')$, the pairs $(\Rvec,B_s^\pi)$ and
        $(\Rvec,B_y)$ have the same conditional distribution given $M' = m'$.
        Hence
        \beq
            \label{eq:contraction-rac-fixed-message}
            \begin{aligned}
                \sum_s H(R_s \mid \Rvec_{<s},M' = m',B_s^\pi) & = \sum_y \sum_{s \in \mathcal I_y(m')} H(R_s \mid \Rvec_{<s},M' = m',B_y)                                              \\
                                                              & = \sum_y \sum_{s \in \mathcal I_y(m')} \bigl[H(R_s \mid \Rvec_{<s},M' = m') - I(R_s:B_y \mid \Rvec_{<s},M' = m')\bigr] \\
                                                              & = H(\Rvec \mid M' = m') - \sum_y \sum_{s \in \mathcal I_y(m')} I(R_s:B_y \mid \Rvec_{<s},M' = m')                      \\
                                                              & \geq H(\Rvec \mid M' = m') - \sum_y I(\Rvec:B_y \mid M' = m').
            \end{aligned}
        \eeq
        The second equality expresses each entropy as an entropy without $B_y$
        minus a conditional mutual information. Non-signaling makes the marginal
        of $(\Rvec,M')$, and hence the first entropy, independent of $y$. In the
        subsequent equality, we used the conditional entropy chain rule and the
        fact that the sets $\mathcal I_y(m')$ form a partition of $[n]$. In the
        last line, for each $y$ we have used
        \beq
            \sum_{s \in \mathcal I_y(m')} I(R_s:B_y \mid \Rvec_{<s},M' = m')
            \leq \sum_{s=0}^{n-1}I(R_s:B_y \mid \Rvec_{<s},M' = m')
            = I(\Rvec:B_y \mid M' = m'). \label{eq:contraction-rac-chain-rule}
        \eeq
        The inequality adds nonnegative conditional mutual informations, and the
        equality is the chain rule. Averaging
        Eq.~\eqref{eq:contraction-rac-fixed-message} over $m'$ gives
        \beq
            \sum_s H(R_s \mid \Rvec_{<s},M',B_s^\pi)
            \geq H(\Rvec \mid M')
                 - \sum_y I(\Rvec:B_y \mid M'). \label{eq:contraction-rac-averaged}
        \eeq
        Consequently,
        \beq
            \begin{aligned}
                \sum_s I(R_s:M',B_s^\pi \mid \Rvec_{<s}) & = \sum_s H(R_s \mid \Rvec_{<s}) - \sum_s H(R_s \mid \Rvec_{<s},M',B_s^\pi) \\
                                                         & = H(\Rvec) - \sum_s H(R_s \mid \Rvec_{<s},M',B_s^\pi)                      \\
                                                         & \leq I(\Rvec:M') + \sum_y I(\Rvec:B_y \mid M')                             \\
                                                         & \leq I(\Rvec,A:M').
            \end{aligned} \label{eq:contraction-correlated-proof}
        \eeq
        Here we used the entropy chain rule,
        Eq.~\eqref{eq:contraction-rac-averaged}, and
        Eq.~\eqref{eq:contraction-conditional-information}. Together with
        Eq.~\eqref{eq:contraction-decoding-data-processing}, this proves the IC
        inequality for arbitrarily correlated entries. In the special case of
        independent entries, $I(R_s:G_s) \leq I(R_s:G_s \mid \Rvec_{<s})$, the
        standard IC inequality follows as well.

\subsection*{Proof of Proposition~\ref{prop:contraction-segment}}
 We divide the proof
    in two parts: first, we prove the exclusion of $P_\lambda$ from $Q_c$ for
    all $\lambda \in (0,1]$, then we calculate the contraction coefficients of
    the corresponding channels $W_{xy}(\lambda)$.

    \begin{proof}[Proof of postquantumness for Proposition~\ref{prop:contraction-segment}]
        Assume, for contradiction, that $P_\lambda$ has a commuting-operator
        quantum realization. Without loss of generality, let $|\psi\rangle$ be a
        unit vector and let $\{A_{a|x}\}_a$ and $\{B_{b|y}\}_b$ be binary
        projective measurements on a common Hilbert space, commuting across
        Alice and Bob, and
        $p(a,b|x,y) = \langle\psi|A_{a|x}B_{b|y}|\psi\rangle$. 
        Consider the four vectors
        \beq
            \label{eq:contraction-four-vectors}
            |u_0\rangle = A_{0|0}|\psi\rangle, \quad |u_1\rangle = A_{0|1}|\psi\rangle, \quad |v_0\rangle =  B_{1|0}|\psi\rangle, \quad |v_1\rangle = B_{1|1}|\psi\rangle.
        \eeq
From the probability table~\eqref{eq:contraction-segment-table}, we find that
        \begin{equation}
            \label{eq:contraction-four-marginals}
            p_A(0|0) = p_A(0|1)
                     = p_B(1|0)
                     = p_B(1|1)
                     = \frac{3 + \lambda}{12}
                     = q,
        \end{equation}
        which implies 
        \beq
            \label{eq:contraction-vector-norms}
            \langle u_0|u_0\rangle = \langle u_1|u_1\rangle
                                   = \langle v_0|v_0\rangle
                                   = \langle v_1|v_1\rangle
                                   = q.
        \eeq
        In addition, the four zero probabilities $p(0,1|x,y) = 0$, $x,y \in \{0,1\}$ give
        \beq
            \label{eq:contraction-cross-orthogonality}
            \langle u_x|v_y\rangle = \langle\psi|A_{0|x}B_{1|y}|\psi\rangle
                                   = p(0,1|x,y)
                                   = 0.
        \eeq
        Furthermore, since $A_{a|x}$ and $B_{b|y}$ are commuting projections,
        their product is itself a projection. Hence
        $p(a,b|x,y) = \langle\psi|A_{a|x}B_{b|y}|\psi\rangle = \left\| A_{a|x}B_{b|y}|\psi\rangle \right\|^2$.
        Therefore,
        \beq
            p(a,b|x,y)
            = 0 \quad
            \Longrightarrow \quad A_{a|x}B_{b|y}|\psi\rangle = 0. \label{eq:contraction-zero-implication}
        \eeq
        We thus find that
        $p(0,0|0,2) = 0$ and $p(0,1|1,2) = 0$ give $A_{0|0}B_{0|2}|\psi\rangle = 0$ and
        $A_{0|1} B_{1|2}|\psi\rangle = 0$. Using commutativity, these mean
        $B_{0|2}|u_0\rangle = 0$ and $B_{1|2}|u_1\rangle = 0$ respectively.
        Hence
        \beq
            \label{eq:contraction-alice-orthogonality}
            \langle u_0|u_1\rangle = \langle u_0|(B_{0|2} + B_{1|2})|u_1\rangle
                                   = 0.
        \eeq
        Similarly, the zeros $p(1,1|2,0) = 0$ and $p(0,1|2,1) = 0$ imply
        $A_{1|2}B_{1|0}|\psi\rangle = 0$ and
        $A_{0|2}B_{1|1}|\psi\rangle = 0$. This in turn implies that
        $A_{1|2}|v_0\rangle = 0$ and $A_{0|2}|v_1\rangle = 0$, and hence
        \beq
            \label{eq:contraction-bob-orthogonality}
            \langle v_0|v_1\rangle = \langle v_0|(A_{0|2} + A_{1|2})|v_1\rangle
                                   = 0.
        \eeq
        Thus the four vectors in Eq.~\eqref{eq:contraction-four-vectors} are
        pairwise orthogonal. Moreover, since each vector is obtained by applying a projection to
        $|\psi\rangle$, $\langle\psi|u_x\rangle = \langle u_x|u_x\rangle$ and
        $\langle\psi|v_y\rangle = \langle v_y|v_y\rangle$. Together with~\eqref{eq:contraction-vector-norms} this implies
        \beq
            \label{eq:contraction-norm-contradiction}
            \begin{split}
                0 & \leq \left\| |\psi\rangle - |u_0\rangle - |u_1\rangle - |v_0\rangle - |v_1\rangle \right\|^2            \\
                  & = 1 - \langle u_0|u_0\rangle - \langle u_1|u_1\rangle - \langle v_0|v_0\rangle - \langle v_1|v_1\rangle \\
                  & = 1 - 4q                                                                                                \\
                  & = -\frac{\lambda}{3}.
            \end{split}
        \eeq
        This is impossible for
        $\lambda > 0$, and therefore $P_\lambda$ has no commuting-operator
        quantum realization.
        \end{proof}

For the proof of the contraction bound we will rely on the following channels.
    The binary $Z$-channel is a standard example of a discrete memoryless
    channel~\cite{polyanskiy:2025}. For $\delta \in [0,1]$, its transition
    matrix is
    \beq
        Z_\delta = \begin{pmatrix}
                       1 & \delta     \\
                       0 & 1 - \delta
                   \end{pmatrix},
    \eeq
    with rows indexed by the output $b$ and columns by the input $a$. It
    transmits the input $0$ without error and maps $1$ to $0$ with probability
    $\delta$. Thus, if $p(A = 0) = q$, then $p(B=0)= q + \delta(1-q)$ and $p(B = 1) = (1 - \delta)(1 - q)$.
    The fixed-input contraction coefficient of $Z_\delta$ is known from~\cite[Eq.~(17)]{anantharam:2013b}, as stated in the following lemma.

    \begin{lemma}[Binary $Z$-channel contraction~\cite{anantharam:2013b}]
        \label{lem:contraction-z-channel}
        Let $q,\delta\in(0,1)$ and $\mathbf q=(q,1-q)^T$ be a fixed binary probability distribution. For every binary probability distribution $\bm w$,
        \beq
        \DKL(Z_\delta\bm w \,\|\, Z_\delta\mathbf q) \leq \eta_Z(\delta,q)\DKL(\bm w \,\|\, \mathbf q),
        \eeq
        where
        \beq
            \eta_Z(\delta,q) = \frac{\log\!\bigl(q + \delta(1 - q)\bigr)}
                       {\log q}. \label{eq:contraction-z-factor}
        \eeq
    \end{lemma}
    In particular, $\eta_Z(\delta,q) = \log p_B(0)/\log p_A(0)$. Therefore, for
    $Z$-channels, this coefficient has an alternative interpretation as the
    ratio of the output and input log-probabilities of the symbol $0$.

    \begin{proof}[Proof of the contraction bound for Proposition~\ref{prop:contraction-segment}]
        Set $q = (3 + \lambda) / 12$, with $0 \leq \lambda \leq 1$. Alice's
        marginals of $P_\lambda$ are
        \beq
            p_A(0|0) = p_A(0|1)
                     = q, \qquad p_A(0|2)
                     = \frac12. \label{eq:contraction-alice-marginals}
        \eeq
        For fixed $(x,y)$, we compute $W_{xy}(b|a) = p(a,b|x,y)/p_A(a|x)$ for
        $P_\lambda$, cf.~Eq.~\eqref{eq:contraction-segment-table} in the main
        text. With rows indexing $b$ and columns indexing $a$, we get
        \beq
            W_{xy} = \begin{pmatrix}
                         1 & \frac{1 - 2q}{1 - q} \\[0.8ex]
                         0 & \frac{q}{1 - q}
                     \end{pmatrix}, \qquad x,y
                   \in \{0,1\}. \label{eq:contraction-channels-01}
        \eeq
        The remaining five channels are
        \beq
            W_{02} = \begin{pmatrix}
                         0 & \frac{1}{2(1 - q)}      \\[0.8ex]
                         1 & \frac{1 - 2q}{2(1 - q)}
                     \end{pmatrix}, \quad W_{12}
                   = \begin{pmatrix}
                         1 & \frac{1 - 2q}{2(1 - q)} \\[0.8ex]
                         0 & \frac{1}{2(1 - q)}
                     \end{pmatrix}, \quad W_{20}
                   = \begin{pmatrix}
                         1 - 2q & 1 \\
                         2q     & 0
                     \end{pmatrix}, \quad W_{21}
                   = \begin{pmatrix}
                         1 & 1 - 2q \\
                         0 & 2q
                     \end{pmatrix}, \quad W_{22}
                   = \begin{pmatrix}
                         \frac12 & \frac12 \\
                         \frac12 & \frac12
                     \end{pmatrix}. \label{eq:contraction-channel-matrices}
        \eeq
        We write $\eta_{xy} = \eta(W_{xy},\bm p_{A|x})$ where $\bm p_{A|x}$ is
        the reference input fixed from the behavior's marginal. We now show that
        these nine matrices are, up to relabelings of the binary outcome
        symbols, $Z$-channels or a constant channel.
        Note that the contraction coefficient is unchanged by relabeling input or output
        symbols, with the reference distribution relabeled along with the
        inputs.

        For $x,y \in \{0,1\}$, $W_{xy}$ is already a $Z$-channel with $p_A(0|x) = q$ and $p_B(0|y) = 1 - q$.
        Lemma~\ref{lem:contraction-z-channel} gives
        \beq
            \eta_{xy} = \frac{\log_2 p_B(0|y)}{\log_2 p_A(0|x)}
                      = \frac{\log_2(1 - q)}{\log_2 q}
                      \eqqcolon \alpha. \label{eq:contraction-alpha-factor}
        \eeq

        For $x = 0,1$ and $y = 2$, swapping Bob's outcome labels $0$ and $1$
        exchanges the rows of $W_{02}$ and gives the $Z$-channel $W_{12}$. Bob's
        uniform marginal is unchanged by this relabeling. For both channels in
        this form, the input reference distribution assigns probability $q$ to
        $a = 0$, and the induced output distribution assigns probability $1/2$
        to $b = 0$. Applying Lemma~\ref{lem:contraction-z-channel} gives
        \beq
            \eta_{02} = \eta_{12}
                      = \frac{\log_2(1/2)}{\log_2 q}
                      \eqqcolon \beta. \label{eq:contraction-beta-factor}
        \eeq

        For $x = 2$ and $y = 0,1$, swapping Alice's outcome labels $0$ and $1$
        exchanges the columns of $W_{20}$ and gives the $Z$-channel $W_{21}$.
        Alice's uniform reference distribution is unchanged by this relabeling.
        For both channels in this form, the input reference distribution assigns
        probability $1/2$ to $a = 0$, and the induced output distribution
        assigns probability $1 - q$ to $b = 0$. Applying
        Lemma~\ref{lem:contraction-z-channel} gives
        \beq
            \eta_{20} = \eta_{21}
                      = \frac{\log_2(1 - q)}{\log_2(1/2)}
                      \eqqcolon \gamma. \label{eq:contraction-gamma-factor}
        \eeq

        Finally, $W_{22}$ is constant: Bob's distribution is independent of
        Alice's outcome. Any reweighting of Alice's outcomes therefore leaves
        Bob's distribution unchanged, and we may take $\eta_{22} = 0$.
        Collecting the four cases, an admissible contraction-factor matrix is
        \beq
            \label{eq:contraction-coefficient-matrix}
            (\eta_{xy}) = \begin{pmatrix}
                              \alpha & \alpha & \beta \\
                              \alpha & \alpha & \beta \\
                              \gamma & \gamma & 0
                          \end{pmatrix},
        \eeq
        for the $\alpha$, $\beta$ and $\gamma$ in
        Eqs.~\eqref{eq:contraction-alpha-factor},
        \eqref{eq:contraction-beta-factor} and
        \eqref{eq:contraction-gamma-factor}. We can now verify the row condition
        $\sum_y \eta_{xy} \le 1$ of Theorem~\ref{thm:contraction-criterion}. For
        the first two rows,
        \beq
            \label{eq:contraction-first-row-condition}
            \begin{aligned}
                2\alpha + \beta \leq 1 & \iff 2\frac{\log_2(1 - q)}{\log_2 q} + \frac{\log_2(1/2)}{\log_2 q} \leq 1 \\
                                       & \iff 2\log_2(1 - q) + \log_2(1/2) \geq \log_2 q                            \\
                                       & \iff (1 - q)^2 \geq 2q                                                     \\
                                       & \iff q \leq 2 - \sqrt3.
            \end{aligned}
        \eeq
        The second equivalence reverses the inequality upon multiplication by
        $\log_2 q < 0$, and the last uses $q \in [1/4,1/3]$. For the third row,
        \beq
            \label{eq:contraction-third-row-condition}
            2\gamma \leq 1 \iff q \leq 1 - \frac1{\sqrt2}.
        \eeq
        The first of the two conditions is the stronger one. Recalling that
        $q = (3 + \lambda) / 12$, it is equivalent to
        \beq
            \label{eq:contraction-lambda-threshold}
            \lambda \leq 21 - 12\sqrt3 \approx 0.2154.
        \eeq
        Hence, for every $0 \leq \lambda \leq 21 - 12\sqrt3$, the hypotheses of
        Theorem~\ref{thm:contraction-criterion} are satisfied. In particular,
        this includes $\lambda = 1/5$ (corresponding to $q = 4/15$), and
        therefore $P^* = P_{1 / 5}$ satisfies information causality for every
        length $n$ and every single-copy protocol.
    \end{proof}

    \section{Multi-copy extension of the contractivity argument}
    \label{app:contraction-multi-copy}

    In this Appendix we prove Lemma~\ref{lem:mc-information-bound}
    and Theorem~\ref{thm:contraction-mixed}. We use the notation introduced in Section~\ref{sec:insufficiency-general}.

    \subsection*{Proof of Lemma~\ref{lem:mc-information-bound}}

    Fix $n,c$ and for simplicity consider a deterministic protocol $\pi_c$ in which Alice 
    uses the copies in the order $1,\ldots,c$. Then, $\ell_i = i$ and
    $x_i = f_{X_i}(\rvec,\mathbf a_{<i})$, where
    $\mathbf a_{<i} = (a_1,\ldots,a_{i-1})$. The more general case will be commented later.
        
    Start from $q_c(\rvec,\mathbf a,m') := p_{\pi_c}^{(P,\ldots, P)}(\rvec,\mathbf a,m')$, 
    which is a marginal of the distribution given in Eq.~\eqref{eq:mc-mixed-protocol} for the specific boxes and protocol considered here.
    To construct from this a distribution $Q(\rvec,m',\bvec_{\mathrm{tot}})$ let us define
    \beq
            q_{c-1}(\bvec_c|\rvec,\mathbf a_{<c},m') \coloneqq \prod_y\left[ \sum_{a_c}W_{x_c y}(b_{c|y}|a_c) q_c(a_c|\rvec,\mathbf a_{<c},m') \right], \label{eq:mc-first-replacement}
    \eeq
        which corresponds to the construction of $(b_{c|1}, \ldots, b_{c|m})$ from $a_c$, yielding the distribution
        \beq
            q_{c-1}(\rvec,\mathbf a_{<c},m',\bvec_c) = q_c(\rvec,\mathbf a_{<c},m') \prod_y\left[ \sum_{a_c}W_{x_c y}(b_{c|y}|a_c) q_c(a_c|\rvec,\mathbf a_{<c},m') \right].
        \eeq
        Notice that by construction the $B_{c|y}$ for different $y$ are
        conditionally independent given $(\Rvec,\mathbf A_{<c},M')$ in $Q_{c-1}$. Now, starting from $Q_{c-1}(\Rvec,\mathbf A_{<c},m',\Bvec_{\geq c})$, we recursively construct distributions $Q_{i-1}(\Rvec,\mathbf A_{<i},m',\Bvec_{\geq i})$ until we reach $Q:=Q_0$. In particular, in step $i$, this gives
        \beq
           q_{i-1}(\rvec,\mathbf a_{<i},m',\bvec_{\geq i}) \coloneqq q_i(\rvec,\mathbf a_{<i},m',\bvec_{>i}) \prod_y\left[ \sum_{a_i}W_{x_i y}(b_{i|y}|a_i) q_i(a_i|\rvec,\mathbf a_{<i},m',\bvec_{>i}) \right]. \label{eq:mc-replacement}
        \eeq
        Notice that by construction the $B_{i|y}$ for different $y$ are again
        conditionally independent given $(\Rvec,\mathbf A_{<i},M', \Bvec_>i)$ in $Q_{i-1}$.
        Furthermore, the
        marginal of $Q_i$ on $(\Rvec,\mathbf A_{\leq i},M')$ for any $i$ remains $q_i(\rvec, a_{\leq i}, m') = p_{\pi_c}^{(P,\ldots, P)}(\rvec, a_{\leq i}, m')$. 
        Therefore, marginalising over $m'$ and $\bvec_{>i}$ 
         gives
        \beq
            \begin{aligned}
                q_{i-1}(B_{i|y} = b|\rvec,\mathbf a_{<i}) & = \sum_{m',\bvec_{>i}} \,
                \sum_{b_{i|y'}:y'\neq y}
                \, q_{i-1}(m',\bvec_{\geq i}|\rvec,\mathbf a_{<i})                                     \\
                                                          & = \sum_{m',\bvec_{>i}} q_i(m',\bvec_{>i}|\rvec,\mathbf a_{<i}) \sum_{a_i} W_{x_i y}(b|a_i) q_i(a_i|\rvec,\mathbf a_{<i},m',\bvec_{>i}) \\
                                                          & = \sum_{a_i} W_{x_i y}(b|a_i) q_i(a_i|\rvec,\mathbf a_{<i})                                                                            \\
                                                          & = \sum_{a_i}W_{x_i y}(b|a_i)p_A(a_i|x_i) = p_B(b|y),
            \end{aligned} \label{eq:mc-bob-marginal}
        \eeq
        which is in particular independent of $\mathbf r, \mathbf a_{<i}$.

        We now estimate the entropy gained in each recursive step. For this, notice that summing Eq.~\eqref{eq:relative-entropy-criterion} over $y$ and using
    Eq.~\eqref{eq:contraction-row-sums} gives
    \beq
        \sum_y \DKL\left( W_{xy}\bm w \,\middle\|\,\bm p_{B|y} \right) \leq \DKL\left( \bm w \,\middle\|\,\bm p_{A|x} \right) \quad\text{for every }x\text{ and every }\bm w, \label{eq:mc-one-copy}
    \eeq
    where $\bm p_{B|y} = W_{xy}\bm p_{A|x}$.    
        Apply
        Eq.~\eqref{eq:mc-one-copy} at $x_i = f_{X_i}(\rvec,\mathbf a_{<i})$ with
        $w(a_i) = q_i(a_i|\rvec,\mathbf a_{<i},m',\bvec_{>i})$, and average over
        $(\rvec,\mathbf a_{<i},m',\bvec_{>i})$. The averaged right-hand side is 
        \beq
            \begin{aligned}
                \sum_{\rvec,\mathbf a_{<i},m',\bvec_{>i}} q_i(\rvec,\mathbf a_{<i},m',\bvec_{>i}) \DKL\left( \bm w\,\middle\|\,\bm p_{A|x_i} \right) 
                &= H(A_i|\Rvec,\mathbf A_{<i})
                  - H_{q_i}(A_i|\Rvec,\mathbf A_{<i},M',\Bvec_{>i}),
            \end{aligned} \label{eq:mc-averaged-input-divergence}
        \eeq
        since $q_i(a_i|\rvec,\mathbf a_{<i}) = p_A(a_i|x_i)$. By
        Eq.~\eqref{eq:mc-bob-marginal}, the averaged left-hand side equals 
        \beq
            \begin{aligned}
               \sum_y \sum_{\rvec,\mathbf a_{<i},m',\bvec_{>i}} q_i(\rvec,\mathbf a_{<i},m',\bvec_{>i}) \DKL\left( W_{x_iy}\bm w\,\middle\|\,\bm p_{B|y} \right) 
                &=h - \sum_y H_{q_{i-1}}(B_{i|y}|\Rvec,\mathbf A_{<i},M',\Bvec_{>i})\\
                &= h - H_{q_{i-1}}(\Bvec_i|\Rvec,\mathbf A_{<i},M',\Bvec_{>i}),
            \end{aligned} \label{eq:mc-averaged-output-divergence}
        \eeq
        where $h
        \coloneqq \sum_y H(B_y)
        = \sum_y H\bigl(p_B(\cdot|y)\bigr)$ and we use that $Q_i$ and $Q_{i-1}$ have the same marginal on $(\Rvec,\mathbf A_{<i},M',\Bvec_{>i}$). The last equality uses the conditional independence in
        Eq.~\eqref{eq:mc-replacement}. Since $Q_i$ and $Q_{i-1}$ have the same
        marginal on $(\Rvec,\mathbf A_{<i},M',\Bvec_{>i})$, subtracting its
        entropy from both sides of the averaged inequality and using the chain
        rule further gives
        \beq
            h - H_{q_{i-1}}(\Rvec,\mathbf A_{<i},M',\Bvec_i,\Bvec_{>i})
            \leq H(A_i|\Rvec,\mathbf A_{<i})
                 - H_{q_i}(\Rvec,\mathbf A_{\leq i},M',\Bvec_{>i}). \label{eq:mc-entropy-step}
        \eeq
        Rearranging and summing over $i = 1,\ldots,c$ yields
        \beq
            \sum_{i=1}^{c}\left[ H_{q_{i-1}}(\Rvec,\mathbf A_{<i},M',\Bvec_{>i-1}) - H_{q_i}(\Rvec,\mathbf A_{\leq i},M',\Bvec_{>i}) \right]
            \geq ch
                 - \sum_{i=1}^{c}H(A_i|\Rvec,\mathbf A_{<i}). \label{eq:mc-telescoping-sum}
        \eeq
        Notice that in the sum on the left all terms except for the
        entropy of $q_0$ and of $q_c$ cancel, which equal
        $H_Q(\Rvec,M',\Bvec_{\mathrm{tot}})$ and $H(\Rvec,\mathbf A,M')$
        respectively. Therefore,
        \beq
            \begin{aligned}
                H_Q(\Rvec,M',\Bvec_{\mathrm{tot}}) & \geq H(\Rvec,\mathbf A,M') + ch - \sum_{i=1}^{c}H(A_i|\Rvec,\mathbf A_{<i}) \\
                                                   & = H(\Rvec,\mathbf A,M') + ch - H(\mathbf A|\Rvec)                           \\
                                                   & = H(\Rvec) + H(M'|\Rvec,\mathbf A) + ch,
            \end{aligned} \label{eq:mc-entropy-lower-bound}
        \eeq
        where the equalities follow from the usual chain rules. Subadditivity of the Shannon entropy furthermore 
        gives
        \beq
            H_Q(M',\Bvec_{\mathrm{tot}})
            \leq H(M') + \sum_{i=1}^{c}\sum_y H_Q(B_{i|y})
            = H(M') + ch, \label{eq:mc-entropy-upper-bound}
        \eeq
        where the last equality uses Eq.~\eqref{eq:mc-bob-marginal}.
        Substituting these bounds into the entropy expression for the mutual
        information gives
        \beq
            \begin{aligned}
                I_Q(\Rvec:M',\Bvec_{\mathrm{tot}}) & = H(\Rvec) + H_Q(M',\Bvec_{\mathrm{tot}}) - H_Q(\Rvec,M',\Bvec_{\mathrm{tot}}) \\
                                                   & \leq H(\Rvec) + H(M') + ch - \bigl[H(\Rvec) + H(M'|\Rvec,\mathbf A) + ch\bigr] \\
                                                   & = I(\Rvec,\mathbf A:M').
            \end{aligned} \label{eq:mc-information-bound}
        \eeq

        It remains to construct the functions $F_s$. For a fixed input $y_i$,
        summing Eq.~\eqref{eq:mc-replacement} over the other coordinates of
        $\Bvec_i$ gives
        \beq
            \begin{aligned}
                q_{i-1}(\rvec,\mathbf a_{<i},m',\bvec_{>i},B_{i|y_i} = b_i) & = q_i(\rvec,\mathbf a_{<i},m',\bvec_{>i}) \sum_{a_i}W_{x_i y_i}(b_i|a_i) q_i(a_i|\rvec,\mathbf a_{<i},m',\bvec_{>i}) \\
                                                                            & = \sum_{a_i}W_{x_i y_i}(b_i|a_i) q_i(\rvec,\mathbf a_{\leq i},m',\bvec_{>i}).
            \end{aligned} \label{eq:mc-selected-coordinate}
        \eeq
        In the first equality, we used normalization of every factor with
        $y \ne y_i$ in Eq.~\eqref{eq:mc-replacement}, after summing over its
        outcome $b_{i|y}$, 
        in the second, we used the definition of conditional
        probability. Applying this identity successively for $i = c,\ldots,1$
        gives, for every fixed $\mathbf y = (y_1,\ldots,y_c)$,
        \beq
            Q(\rvec,m',B_{1|y_1} = b_1,\ldots,B_{c|y_c} = b_c) = \sum_{\mathbf a} p_{\pi_c}^{(P,\ldots, P)}(\rvec,\mathbf a,m') \prod_{i=1}^{c}W_{x_i y_i}(b_i|a_i).\label{eq:mc-outcome-distrib}
        \eeq
        The right-hand side is the original joint distribution for Bob, when his inputs
        equal $\mathbf y$. 
    This shows that the distribution $Q$ contains Bob's outcome probabilities for any $\mathbf y$ he might choose simultaneously, as long as each box is probed only once.

    Any realization $\bvec_{\rm tot}$ of the auxiliary variables $\Bvec_{\rm tot}$ contains one entry $b_{j|y}$ for each copy $j$ and input $y$. For fixed $s,m'$, let us consider Bob's functions $f_{J_k}, f_{Y_k}$:
        whenever they specify copy $j$ and input $y$, we look at the value
        $b_{j|y}$ and use it as his outcome to determine the subsequent
        choices of $f_{J_k}, f_{Y_k}$, and we finally apply the
        guessing function $f_G$.  This defines
        \beq
            \widetilde g_s \coloneqq F_s(m',\bvec_{\mathrm{tot}}). \label{eq:mc-auxiliary-function}
        \eeq
        To verify that $F_s$ has the property stated in the lemma, fix $s,m'$
        and a list $\bvec$ of Bob's $c$ outcomes. The functions
        $\{f_{J_k}\},\{f_{Y_k}\}$ determine which inputs $\mathbf y$ he would
        use if these outcomes occurred. Substitute these inputs into
        Eq.~\eqref{eq:mc-outcome-distrib}. Its left-hand side is the probability
        that the auxiliary lookup table supplies these outcomes, jointly with
        $(\rvec,m')$; its right-hand side is the corresponding probability in
        the original protocol. Thus, $F_s$ has the desired property.

    The case where Alice chooses adaptively which copy to use next using
    \beq
            \ell_i = f_{L_i}(\rvec,a_{\ell_1}, \dots, a_{\ell_{i-1}}), \quad x_{\ell_i}
                   = f_{X_i}(\rvec, a_{\ell_1}, \dots, a_{\ell_{i-1}}),
    \eeq
    is similar, but the substitutions happen in the ``reverse chronological order''
    $(\widetilde{A}_{c}, \ldots,\widetilde{A}_{1}) \coloneqq (A_{\ell_c}, \dots, A_{\ell_1})$. More precisely, fix a value of
    $(\rvec, \mathbf a, m')$, and calculate
    $\ell_1 = f_{L_1}(\rvec)$, $\widetilde a_1 \coloneqq a_{\ell_1}$,
    and then recursively
    \beq
    \begin{aligned}
        \ell_i = f_{L_i}(\rvec,\widetilde a_1, \dots, \widetilde a_{i-1}), \qquad \widetilde a_i \coloneqq a_{\ell_i}
    \end{aligned}
    \eeq
    for all $i$. This lets us write the starting joint distribution as
    $q_c(\rvec, \widetilde{\mathbf a},m') = p_{\pi_c}(\rvec,\widetilde{\mathbf a},m')$, with Alice's outcomes in chronological
    order. Starting from this distribution, we follow the exact same construction as the
    lemma, with $\widetilde{\mathbf A}$ instead of $\mathbf A$ (we replace 
    $\widetilde A_c$ by $\mathbf B_{\ell_c}$, then $\widetilde A_{c-1}$ by $\mathbf B_{\ell_{c-1}}$ etc.). The entropy proofs still go
    through identically with $\widetilde A_i$ in place of $A_i$ everywhere,
    and in the final bounds, one has
    $I(\mathbf R , \widetilde{\mathbf A} : M')$ which is equal to
    $I(\mathbf R , \mathbf A : M')$. 
    Notice that the distribution of $\mathbf A $ and $\widetilde{\mathbf A} $ and thus their entropy may differ, however this reordering is deterministic (conditioned on $\mathbf r$) and invertible, and thus does not change their joint entropy with $\Rvec$.
    The construction and verification of $F_s$ extend similarly to the previous case: whenever
    Bob chooses copy $j$ and input $y$, one can use $b_{j|y}$ as his outcome. At the step replacing $\widetilde{A}_i$,
    the variables $(\rvec,\widetilde a_1, \dots, \widetilde a_{i-1})$ fix $\ell_i$ and $x_{\ell_i}$, 
    so the same calculation as Eq.~\eqref{eq:mc-selected-coordinate} applies. Repeating this gives an analogue of Eq.~\eqref{eq:mc-outcome-distrib} for every fixed $\mathbf y$, and the verification of $F_s$ proceeds as the previous case.
    
    \subsection*{
        \texorpdfstring{Protocols with $P^*$ and local behaviors}
        {Protocols with P* and local behaviors}
    }

    Theorem~\ref{thm:contraction-multi-copy} establishes IC for arbitrary copies
    of $P^*$. We now extend this to the case of $c$ possibly different boxes,
    each taken from $\mathcal S = \conv(\LHV \cup \{P^*\})$, by showing that
    mixing one behavior while keeping the others fixed cannot ``activate'' a
    violation of IC (Lemma~\ref{lem:mc-mixing}), together with the fact that
    local deterministic boxes can be simulated locally
    (Lemma~\ref{lem:mc-deterministic-boxes}).

    For the general protocol $\pi$ in Eq.~\eqref{eq:mc-mixed-protocol}, define
    the IC deficit
    \beq
        \Delta(\mathbf P,\pi) = \sum_{s=0}^{n-1}I(R_s:G_s \mid \Rvec_{<s})
                                - I(\Rvec,\mathbf A:M'). \label{eq:mc-deficit}
    \eeq
    We say that the set of boxes $\mathbf P$ satisfies IC if
    $\Delta(\mathbf P,\pi) \leq 0$ for every database length, every database
    distribution, and every deterministic protocol.

    \begin{lemma}
        \label{lem:mc-mixing}
        Let $\mathbf P'$ and $\mathbf P''$ differ only in one system, at position
        $i$:
        \beq
            \begin{aligned}
                \mathbf P'  & = P_1 \otimes \cdots \otimes P' \otimes \cdots \otimes P_c,  \\
                \mathbf P'' & = P_1 \otimes \cdots \otimes P'' \otimes \cdots \otimes P_c.
            \end{aligned} \label{eq:mc-mixing-tuples}
        \eeq
        If both $\mathbf P'$ and $\mathbf P''$ satisfy IC, so does
        \beq
            \mathbf P = P_1
                        \otimes \cdots
                        \otimes [\lambda P' + (1 - \lambda)P'']
                        \otimes \cdots
                        \otimes P_c, \quad 0
                      \leq \lambda
                      \leq 1. \label{eq:mc-mixed-tuple}
        \eeq
    \end{lemma}
    \begin{proof}
        Fix a database distribution $p(\rvec)$ and a protocol $\pi_c$. Each box
        occurs once in the product in Eq.~\eqref{eq:mc-mixed-protocol}, and
        $f_G$ is fixed. The distributions over all variables therefore satisfy
        $p = \lambda p' + (1 - \lambda)p''$. Since $p(s)$ is fixed, the same
        mixture holds for $p(g_s,r_s,\rvec_{<s}) = p(g,r_s,\rvec_{<s} | S = s)$.
        Dividing by the common marginal $p(r_s,\rvec_{<s})$ gives
        \beq
            \begin{aligned}
                p(g_s|r_s,\rvec_{<s}) & = \frac{p(g_s,r_s,\rvec_{<s})}{p(r_s,\rvec_{<s})}                        \\
                                      & = \lambda p'(g_s|r_s,\rvec_{<s}) + (1 - \lambda)p''(g_s|r_s,\rvec_{<s}).
            \end{aligned} \label{eq:mc-mixing-guess-channel}
        \eeq
        For each $\rvec_{<s}$, mutual information is convex in this channel
        because the input distribution $p(r_s|\rvec_{<s})$ is fixed. Averaging
        with the common weights $p(\rvec_{<s})$ gives
        \beq
            I(R_s:G_s|\Rvec_{<s})
            \leq \lambda I'(R_s:G_s \mid \Rvec_{<s})
                 + (1 - \lambda)I''(R_s:G_s \mid \Rvec_{<s}). \label{eq:mc-mixing-information-left}
        \eeq
        Conversely, the channel $p(m'|\rvec,\mathbf a)$ is the same in the three
        cases, while its input distribution satisfies
        $p(\rvec,\mathbf a) = \lambda p'(\rvec,\mathbf a) + (1 - \lambda)p''(\rvec,\mathbf a)$.
        Mutual information is concave in the input distribution for a fixed
        channel, hence
        \beq
            I(\Rvec,\mathbf A:M')
            \geq \lambda I'(\Rvec,\mathbf A:M')
                 + (1 - \lambda)I''(\Rvec,\mathbf A:M'). \label{eq:mc-mixing-information-right}
        \eeq
        Subtracting gives
        \beq
            \Delta(\mathbf P,\pi) \leq \lambda\Delta(\mathbf P',\pi)
                                       + (1 - \lambda)\Delta(\mathbf P'',\pi)
                                  \leq 0. \label{eq:mc-mixing-deficit}
        \eeq
        Since $p(\rvec)$ and $\pi_c$ were arbitrary, $\mathbf P$ satisfies IC.
    \end{proof}

    The next lemma shows that any multi-copy protocol involving copies of a
    (possibly nonlocal) behavior $P^*$ and local deterministic boxes can be seen as a
    protocol on copies of $P^*$ only where the local deterministic boxes are
    simulated locally.

    \begin{lemma}%
        \label{lem:mc-deterministic-boxes}
        Let $\mathbf P = P_1 \otimes \cdots \otimes P_c$, where $P_i = D_i$ for
        $i$ in a set $\mathcal D$ of $k < c$ indices, with
        \beq
            D_i(a,b|x,y) = \delta_{a,\alpha_i(x)}\delta_{b,\beta_i(y)}, \label{eq:mc-deterministic-resource}
        \eeq
        and all other boxes are $P^*$. Every deterministic protocol on
        $\mathbf P$ is equivalent to a deterministic $(c - k)$-copy protocol on
        $P^{* \otimes (c-k)}$ with the same IC deficit $\Delta$.
    \end{lemma}
    \begin{proof}
        First suppose that Alice uses the copies in the order $1,\ldots,c$. The
        idea is that the parties can calculate the outcomes of every $D_i$
        locally, so these resources can be removed. Write
        \beq
            \mathbf a_*
            = (a_i)_{i \notin \mathcal D},\qquad \bvec_*
            = (b_i)_{i \notin \mathcal D} \label{eq:mc-retained-outcomes}
        \eeq
        for the outcomes on the $P^*$ resources. Alice follows her original
        strategy $\{f_{X_i}\}$ and generates
        \beq
            x_i
            = f_{X_i}(\rvec,\mathbf a_{<i}),\qquad a_i
            = \begin{cases}
                  \alpha_i(x_i),                                               & i \in \mathcal D,    \\
                  \text{outcome of copy }i\text{ of }P^* \text{ at input }x_i, & i \notin \mathcal D.
              \end{cases} \label{eq:mc-local-simulation}
        \eeq
        Thus $\mathbf a_{<i}$, and hence $x_i$, are functions of $\rvec$ and of
        her previous outcomes on the $P^*$ copies. At the end, the full outcome
        vector can be written as a function
        $\mathbf a = \phi(\rvec,\mathbf a_*)$. The reduced protocol uses the
        channel
        \beq
            \widetilde p(m'|\rvec,\mathbf a_*) \coloneqq p(m'|\rvec,\phi(\rvec,\mathbf a_*)). \label{eq:mc-reduced-message-channel}
        \eeq
        Bob follows his decision tree, but whenever it asks for a copy
        $j \in \mathcal D$ at input $y$, he calculates $b_j = \beta_j(y)$. His
        next $P^*$ copy, its input, and his guess are then functions only of
        $(s,m')$ and his previous outcomes on the $P^*$ copies.

        These rules define a deterministic $(c - k)$-copy protocol
        $\widetilde\pi$ on $P^{* \otimes (c-k)}$. In the global distribution of
        Eq.~\eqref{eq:mc-mixed-protocol}, the deterministic resources contribute
        only $\delta_{a_i,\alpha_i(x_i)}\delta_{b_i,\beta_i(y_i)}$. Summing over
        their outcomes substitutes these local functions into the remaining
        factors, giving precisely the distribution of $\widetilde\pi$. The two
        protocols therefore have the same joint distribution of $(\Rvec,S,M',G)$
        and hence the same conditional information terms.

        For the right-hand side, $\mathbf A = \phi(\Rvec,\mathbf A_*)$ and
        $\mathbf A_*$ is a subvector of $\mathbf A$. Hence $(\Rvec,\mathbf A)$
        and $(\Rvec,\mathbf A_*)$ determine each other, and data processing in
        both directions gives
        \beq
            I(\Rvec,\mathbf A:M') = I(\Rvec,\mathbf A_*:M'). \label{eq:mc-local-removal-cost}
        \eeq
        The IC deficits are therefore equal.

        Alice's adaptive copy order is treated as in the previous results: use
        the $P^*$ outcomes in chronological order and compute each local
        deterministic outcome when needed. Alice's complete list of outputs is
        still determined by $\rvec$ and the outcomes of the non-deterministic
        boxes, so $\Delta$ stays the same.
    \end{proof}

    We can now combine these lemmas with
    Theorem~\ref{thm:contraction-multi-copy} to prove
    Theorem~\ref{thm:contraction-mixed}.
    \begin{proof}[Proof of Theorem~\ref{thm:contraction-mixed}]
        Every element of $\mathcal S$ is a convex combination of $P^*$ and local
        deterministic behaviors. We first consider every tuple
        $\mathbf Q = Q_1 \otimes \cdots \otimes Q_c$ in which each $Q_i$ is
        either $P^*$ or local deterministic. Let $k < c$ be the number of local
        deterministic boxes (in the $k=c$ case, IC is trivially satisfied as all 
        resources are local). Lemma~\ref{lem:mc-deterministic-boxes} gives a
        protocol on $P^{* \otimes (c-k)}$ with the same IC deficit, which is
        nonpositive by Theorem~\ref{thm:contraction-multi-copy}.
        Lemma~\ref{lem:mc-mixing} easily extends to any finite convex
        combination $\sum_j \lambda_j P^{(j)}$ on position $i$. Applying it one
        position at a time therefore proves the claim for all
        $P_1,\ldots,P_c \in \mathcal S$:
        \beq
            Q_1 \otimes \cdots \otimes Q_c \;
            \implies \; P_1 \otimes Q_2 \otimes \cdots \otimes Q_c \;
            \implies \; \cdots\;
            \implies \; P_1
                     \otimes \cdots
                     \otimes P_c. \label{eq:mc-mixing-one-position}
        \eeq
        In other words, any set of the form $Q_1 \otimes \cdots \otimes Q_c$
        satisfies IC, thus by the lemma any set of the form
        $P_1 \otimes Q_2 \otimes \cdots \otimes Q_c$ satisfies IC. Applying the
        lemma again, sets of the form
        $P_1 \otimes P_2 \otimes \cdots \otimes Q_c$ satisfy IC, and so on.
    \end{proof}

    \section{Other formulations of information causality}
    \label{app:other-IC-formulations}

    In this appendix, we prove that the set $\mathcal S$ of
    Theorem~\ref{thm:contraction-mixed} also satisfies the other IC inequalities
    proposed in the literature, namely those of Chaves, Majenz and
    Gross~\cite{chaves:2015c}, Pollyceno et al.~\cite{pollyceno:2026}, and Yu
    and Scarani~\cite{yu:2026}. The first two follow from
    Eq.~\eqref{eq:contraction-mixed-ic} of the main text with $(M',G_s)$ in
    place of $G_s$. Indeed, in Lemma~\ref{lem:mc-information-bound},
    $(M',\widetilde G_s)$ is also a function of $(M',\Bvec_{\mathrm{tot}})$ and
    has the same joint distribution with $\Rvec$ as $(M',G_s)$. The conditional
    data-processing step in Eq.~\eqref{eq:mc-correlated-ic} therefore applies to
    this pair. The local simulation and mixing arguments of
    Lemmas~\ref{lem:mc-deterministic-boxes} and~\ref{lem:mc-mixing} also apply
    with $(M',G_s)$ as the output.

    \begin{proposition}[Generalisation of Theorem~\ref{thm:contraction-mixed} to IC from~\cite{chaves:2015c}]
        \label{prop:of-cmg}
        Every protocol using $c$ independent resources
        from $\mathcal S$ and a database of $n$ possibly correlated entries
        satisfies
        \beq
            \sum_{s=0}^{n-1}I(R_s:M',G_s) + \sum_{s=1}^{n-1}I(R_0:R_s \mid M',G_s)
            \leq I(M:M')
                 + \sum_{s=0}^{n-1}H(R_s)
                 - H(\Rvec). \label{eq:of-cmg-noisy}
        \eeq
        For a noiseless channel, $M' = M$ and $I(M:M') = H(M)$, recovering the
        inequality of~\cite[Eq.~(7)]{chaves:2015c}.
    \end{proposition}
    \begin{proof}
        For each $s \geq 1$, the chain rule and data processing give
        \beq
            \begin{aligned}
                I(R_s:M',G_s) + I(R_0:R_s \mid M',G_s) & = I(R_s:R_0,M',G_s)                                  \\
                                                       & \leq I(R_s:\Rvec_{<s},M',G_s)                        \\
                                                       & = I(R_s:\Rvec_{<s}) + I(R_s:M',G_s \mid \Rvec_{<s}).
            \end{aligned} \label{eq:of-cmg-entry-bound}
        \eeq
        The inequality uses the fact that $R_0$ is a component of $\Rvec_{<s}$.
        We next sum the terms $I(R_s:\Rvec_{<s})$ over $s = 1,\dots,n - 1$. The
        entropy chain rule gives
        $\sum_{s=1}^{n-1}H(R_s|\Rvec_{<s}) = H(\Rvec) - H(R_0)$, so
        \beq
            \sum_{s=1}^{n-1}I(R_s:\Rvec_{<s})
            = \sum_{s=1}^{n-1}\bigl[H(R_s) - H(R_s|\Rvec_{<s})\bigr]
            = \sum_{s=0}^{n-1}H(R_s)
              - H(\Rvec). \label{eq:of-cmg-correlation-correction}
        \eeq
        Now sum Eq.~\eqref{eq:of-cmg-entry-bound} over $s = 1,\ldots,n - 1$, add
        $I(R_0:M',G_0)$ to both sides, and substitute
        Eq.~\eqref{eq:of-cmg-correlation-correction}. This gives
        \beq
            \begin{aligned}
                \sum_{s=0}^{n-1}I(R_s:M',G_s) + \sum_{s=1}^{n-1}I(R_0:R_s \mid M',G_s) & \leq \sum_{s=0}^{n-1}H(R_s) - H(\Rvec) + \sum_{s=0}^{n-1}I(R_s:M',G_s \mid \Rvec_{<s}) \\
                                                                                       & \leq \sum_{s=0}^{n-1}H(R_s) - H(\Rvec) + I(M:M').
            \end{aligned}
        \eeq
        In the first line, the $s = 0$ term of the last sum is
        $I(R_0:M',G_0)$, since $\Rvec_{<0}$ is empty. The last inequality is
        our general result,
        Eq.~\eqref{eq:contraction-mixed-ic}, with $(M',G_s)$ in place of $G_s$
        and $I(M:M')$ as the right-hand side, which is in general less
        restrictive. This proves Eq.~\eqref{eq:of-cmg-noisy}.
    \end{proof}

    \begin{proposition}[Generalisation of Theorem~\ref{thm:contraction-mixed} to IC from~\cite{pollyceno:2026}]
        \label{prop:of-gc10}
        Every protocol using $c$ independent resources
        from $\mathcal S$ and a two-entry database satisfies
        \beq
            I(R_0:M',G_s) + I(R_1:R_0,M',G_{s'})
            \leq I(M:M') + I(R_0:R_1) \label{eq:of-gc10}
        \eeq
        for all Bob questions $s,s'$. The same inequality holds with $R_0$ and $R_1$
        interchanged. For independent entries, this recovers the inequality
        of~\cite[Eq.~(10)]{pollyceno:2026}.
    \end{proposition}
    \begin{proof}
        Fix $s,s'$ and apply Proposition~\ref{prop:of-cmg} with $n = 2$ to the
        protocol in which Bob follows his original strategy for question $s$ when
        asked for $R_0$, and that for $s'$ when asked for $R_1$. Alice's
        strategy is unchanged. By the chain rule,
        \beq
            I(R_1:M',G_{s'}) + I(R_0:R_1 \mid M',G_{s'}) = I(R_1:R_0,M',G_{s'}).
        \eeq
        The right-hand side of Eq.~\eqref{eq:of-cmg-noisy} is the sum of $I(M:M')$ and
        $H(R_0) + H(R_1) - H(R_0,R_1) = I(R_0:R_1)$, which gives
        Eq.~\eqref{eq:of-gc10}. The other ordering follows by relabeling $R_0$
        and $R_1$.
    \end{proof}

Yu and Scarani have proposed a reformulation of IC that takes into account that different guesses of Bob may contain some shared information about Alice's database~\cite{yu:2026}. They use the measure of
    redundant information~\cite{harder:2013},
    to quantify this. For two outputs $G_1,G_2$, let
    $\mathcal C_i = \conv\{p(\Rvec|G_i = g):p(G_i = g) > 0\}$ and define
    \beq
        I^\pi(G_1 \searrow G_2)
        = I(\Rvec:G_1)
          - \sum_g p(G_1 = g)\min_{\bm v \in \mathcal C_2} \DKL\!\left( p(\Rvec|G_1 = g)\,\middle\|\,\bm v \right). \label{eq:of-projected}
    \eeq
    Define $I^\pi(G_2 \searrow G_1)$ analogously. Their redundant information is
    \beq
        I_{\mathrm{red}}(\Rvec;G_1,G_2) = \min\{I^\pi(G_1 \searrow G_2),I^\pi(G_2 \searrow G_1)\}.
    \eeq
    This quantity depends only on the two marginals $p(\rvec,g_1)$ and
    $p(\rvec,g_2)$, so it does not require any joint $p(\rvec,g_1,g_2)$ to be
    specified. The information causality inequality is then
\beq
    I(\Rvec:G_1) + I(\Rvec:G_2) - I_{\mathrm{red}}(\Rvec;G_1,G_2) \leq C,
\eeq
where $C$ is the channel capacity.
    Extending this formulation beyond two decoding procedures remains
an open problem.

    \begin{proposition}[Generalisation of Theorem~\ref{thm:contraction-mixed} to IC from~\cite{yu:2026}]
        \label{prop:of-ys}
        Every  protocol using $c$ independent resources
        from $\mathcal S$ satisfies
        \beq
            I(\Rvec:G_s) + I(\Rvec:G_{s'}) - I_{\mathrm{red}}(\Rvec;G_s,G_{s'}) \leq I(M:M') \label{eq:of-ys}
        \eeq
        for all questions $s,s'$. Since $I(M:M')$ is at most the channel
        capacity, this implies the IC definition of~\cite[Eq.~(4)]{yu:2026}.
    \end{proposition}
    \begin{proof}
    For any $s \ne s'$, denote Bob's corresponding guesses by $G_1$ and $G_2$
    respectively.
        Lemma~7 of~\cite{harder:2013} bounds $I^\pi(G_1 \searrow G_2)$ from
        below by $I(\Rvec:G_1) + I(\Rvec:G_2) - I(\Rvec:G_1,G_2)$. Exchanging
        $G_1$ and $G_2$ gives the same lower bound for
        $I^\pi(G_2 \searrow G_1)$, hence also for their minimum
        $I_{\mathrm{red}}(\Rvec;G_1,G_2)$. Rearranging gives
        \beq
            I(\Rvec:G_1) + I(\Rvec:G_2) - I_{\mathrm{red}}(\Rvec;G_1,G_2) \leq I(\Rvec:G_1,G_2) \label{eq:of-ys-joint-bound}
        \eeq
        for every joint distribution of $(\Rvec,G_1,G_2)$. Since the left-hand
        side depends only on the two marginals, it suffices to construct a joint
        distribution reproducing $(\Rvec,G_s)$ and $(\Rvec,G_{s'})$, with
        $I(\Rvec:G_1,G_2) \leq I(M:M')$.

        For a set $P_1, \dots, P_c$, that we denote as $\mathbf{P}$ for simplicity, %
        containing only copies of $P^*$ and
        local deterministic boxes, simulate the latter as in
        Lemma~\ref{lem:mc-deterministic-boxes} and use the construction of
        Lemma~\ref{lem:mc-information-bound} for the remaining copies of $P^*$. Both
        auxiliary guesses are functions of $(M',\Bvec_{\mathrm{tot}})$. Their
        joint distribution with $\Rvec$, denoted by $Q_{\mathbf P}$, reproduces
        the distributions of $(\Rvec,G_1)$ and $(\Rvec,G_2)$ in the protocol
        that uses $\mathbf P$, and by data
        processing satisfies
        \beq
            I_{Q_{\mathbf P}}(\Rvec:G_1,G_2) \leq I_{\mathbf P}(M:M'), \label{eq:of-coupling-bound}
        \eeq
        where $I_{\mathbf P}$ is evaluated in the protocol that uses $\mathbf P$. Every product of resources from $\mathcal S$ is a convex combination of 
        such tuples. Let $w_{\mathbf P}$ be its weights and take
        $Q = \sum_{\mathbf P} w_{\mathbf P} Q_{\mathbf P}$. This joint
        distribution reproduces the distributions of $(\Rvec,G_1)$ and 
        $(\Rvec,G_2)$ of the original protocol, and
        \beq
            \begin{aligned}
                I_Q(\Rvec:G_1,G_2) & \leq \sum_{\mathbf P} w_{\mathbf P} I_{Q_{\mathbf P}}(\Rvec:G_1,G_2) \\
                                   & \leq \sum_{\mathbf P} w_{\mathbf P} I_{\mathbf P}(M:M')              \\
                                   & \leq I(M:M').
            \end{aligned}
        \eeq
        The first inequality uses convexity of the mutual information in $Q(g_1,g_2|\rvec)$ 
        at fixed $p(\mathbf r)$. 
        The last uses concavity of the mutual information in $p(m)=\sum_{\mathbf P}w_{\mathbf P} p_{\mathbf P}(m)$, for a fixed communication channel $p(m'|m)$,
        as in
        Lemma~\ref{lem:mc-mixing}. Applying Eq.~\eqref{eq:of-ys-joint-bound} to
        $Q$ proves Eq.~\eqref{eq:of-ys}.
    \end{proof}

    These proofs apply for any behavior $P^*$ satisfying the hypotheses of
    Theorem~\ref{thm:contraction-criterion}. By
    Proposition~\ref{prop:contraction-segment}, this includes $P_\lambda$ with
    $0 < \lambda \leq 21 - 12\sqrt3$, which therefore satisfy all notions of 
    information causality considered here.

\end{widetext}
\end{document}